\documentclass[a4paper,UKenglish]{lipics-v2019}

\usepackage{enumerate}
\usepackage{thm-restate}
\usepackage{com}
\usepackage{thm-restate}

\nolinenumbers

\usepackage{xcolor,pgf,tikz}
\usetikzlibrary{arrows,arrows.meta,automata,shapes,decorations,calc,positioning}

\title{Nash equilibria in games over graphs equipped
  with a communication mechanism}

\titlerunning{Nash equilibria in games over graphs
  equipped with a communication mechanism}

\author{Patricia Bouyer}{LSV, CNRS, ENS Paris-Saclay, Universit\'e
  Paris-Saclay,  France}{bouyer@lsv.fr}{https://orcid.org/0000-0002-2823-091}{}

\author{Nathan Thomasset}{LSV, ENS Paris-Saclay, CNRS, Universit\'e
  Paris-Saclay, France}{nathan.thomasset@ens-paris-saclay.fr}{}{}

\authorrunning{P. Bouyer and N. Thomasset}

\relatedversion{Full version available at
  http://arxiv.org/abs/1906.07753}

\funding{This work has been partly supported by ERC project EQualIS
  (FP7-308087).}

\Copyright{P. Bouyer and N. Thomasset}



\category{}

\relatedversion{}

\supplement{}

\acknowledgements{}

\EventEditors{Peter Rossmanith, Pinar Heggernes, and Joost-Pieter Katoen}
\EventNoEds{3}
\EventLongTitle{44th International Symposium on Mathematical Foundations of Computer Science (MFCS 2019)}
\EventShortTitle{MFCS 2019}
\EventAcronym{MFCS}
\EventYear{2019}
\EventDate{August 26--30, 2019}
\EventLocation{Aachen, Germany}
\EventLogo{}
\SeriesVolume{138}
\ArticleNo{14}

\begin{document}

\maketitle

\begin{abstract}
  We study pure Nash equilibria in infinite-duration games on graphs,
  with partial visibility of actions but communication (based on a
  graph) among the players. We show that a simple communication
  mechanism consisting in reporting the deviator when seeing it and
  propagating this information is sufficient for characterizing Nash
  equilibria.
  We propose an epistemic game construction, which conveniently
  records important information about the knowledge of the
  players. With this abstraction, we are able to characterize Nash
  equilibria which follow the simple communication pattern via winning
  strategies. We finally discuss the size of the construction, which
  would allow efficient algorithmic solutions to compute Nash
  equilibria in the original game.
\end{abstract}

\section{Introduction}

Multiplayer concurrent games over graphs allow to model rich
interactions between players. Those games are played as follows.  In a
state, each player chooses privately and independently an action,
defining globally a move (one action per player); the next state of
the game is then defined as the successor (on the graph) of the
current state using that move; players continue playing from that new
state, and form a(n infinite) play. Each player then gets a reward
given by a payoff function (one function per player). In particular,
objectives of the players may not be contradictory: those games are
non-zero-sum games, contrary to two-player games used for controllers
or reactive synthesis~\cite{thomas02,henzinger05}.

Using solution concepts borrowed from game theory, one can describe
interactions among the players, and in particular rational
behaviours of selfish players. One of the most basic and classically
studied solution concepts is that of Nash
equilibria~\cite{nash50}. A Nash equilibrium is a strategy profile
where no player can improve her payoff by unilaterally changing her
strategy. The outcome of a Nash equilibrium can therefore be seen as a
rational behaviour of the system.
While very much studied by game theorists, e.g. over (repeated) matrix
games, such a concept (and variants thereof) has been only rather
recently studied over infinite-duration games on graphs. Probably the
first works in that direction
are~\cite{CMJ04,CHJ06,ummels06,ummels08}. Several series of works have
followed. To roughly give an idea of the existing results, pure Nash
equilibria always exist in turn-based games for $\omega$-regular
objectives~\cite{UW11} but not in concurrent games, even with simple
objectives; they can nevertheless be
computed~\cite{UW11,BBMU15,brenguier16,bouyer18} for large classes of
objectives. The problem becomes harder with mixed (that is,
stochastic) Nash equilibria, for which we often cannot decide the
existence~\cite{UW11a,BMS14}.

Computing Nash equilibria requires to (i) find a good behaviour of the
system; (ii) detect deviations from that behaviour, and identify
deviating players; (iii) punish them.  This simple characterization of
Nash equilibria is made explicit in~\cite{CFGR16}.  Variants of Nash
equilibria require slightly different ingredients, but they are mostly
of a similar vein.

Many of those results are proven using the construction of a
two-player game, in which winning strategies correspond (in some
precise sense) to Nash equilibria in the original game. This
two-player game basically records the knowledge of the various players
about everything which can be uncertain: (a) the possible deviators
in~\cite{BBMU15}, and (b) the possible states the game can be
in~\cite{bouyer18}.  Extensions of this construction can be used for
other solution concepts like robust equilibria~\cite{brenguier16} or
rational synthesis~\cite{COT18}.

\medskip In this work, we consider infinite-duration games on graphs,
in which the game arena is perfectly known by all the players, but
players have only a partial information on the actions played by the
other players. 

The partial-information setting of this work is inspired
by~\cite{RT98}: it considers repeated games played on matrices, where
players only see actions of their neighbours. Neighbours are specified
by a communication graph. To ensure a correct detection of deviators,
the solution is to propagate the identity of the deviator along the
communication graph. A fingerprint (finite sequence of actions) of
every player is agreed at the beginning, and the propagation can be
made properly if and only if the communication graph is 2-connected,
ensuring large sets of Nash equilibria (formalized as a folk
theorem). Fingerprints are not adapted to the setting of graphs, since
they may delay the time at which a player will learn the identity of
the deviator, which may be prohibitive if a bad component of a graph
is then reached.

We therefore propose to add real communication among
players. Similarly to~\cite{RT98}, a player can communicate only with
her neighbours (also specified by a communication graph), but can send
arbitrary messages (modelled as arbitrary words over alphabet
$\{0,1\}$).
We assume that visited states are known by the players, hence only the
deviator (if any) may be unknown to the players. In this setting, we
show the following results:
\begin{itemize}
\item We show that a very simple epidemic-like communication mechanism
  is sufficient for defining Nash equilibria. It consists in (a)
  reporting the deviator (for the neighbours of the deviator) as soon
  as it is detected, and (b) propagating this information (for the
  other players).
\item We build an epistemic game, which tracks those strategy profiles
  which follow the above simple communication pattern. This is a
  two-player turn-based game, in which \Eve (the first player) suggest
  moves, and \Adam (the second player) complies (to generate the main
  outcome), or not (to mimic single-player deviations). The
  correctness of the construction is formulated as
  follows: there is a Nash equilibrium in the original game of payoff
  $p$ if and only if there is a strategy for \Eve in the epistemic
  game which is winning for $p$.
\item We analyze the complexity of this construction.
\end{itemize}

Note that we do not assume connectedness of the communication graph, hence
the particular case of a graph with no edges allows to recover the
setting of~\cite{BBMU15} while a complete graph allows to recover the
settings of~\cite{UW11,brenguier16}.

\medskip In Section~\ref{sec:defs}, we define our model and give an
example to illustrate the role of the communication graph. In
Section~\ref{sec:reducs}, we prove the simple communication
pattern. In Section~\ref{sec:epistemic}, we construct the epistemic
game and discuss its correctness. 
In Section~\ref{sec:complexity}, we discuss complexity issues.
All proofs are postponed to the technical appendix.

\section{Definitions}
\label{sec:defs}

We use the standard notations $\bbR$ (resp. $\bbQ$, $\bbN$) for the
set of real (resp. rational, natural) numbers.  If $S$ is a subset of
$\bbR$, we write $\overline{S}$ for $S \cup \{-\infty,+\infty\}$.

Let $S$ be a finite set and $R \subseteq S$. If $m$ is an $S$-vector
over some set $\Sigma$, we write $m(R)$ (resp. $m(-R)$) for the vector
composed of the $R$-components of $m$ (resp. all but the
$R$-components). We also use abusively the notations $m(i)$
(resp. $m(-i)$) when $i$ is a single element of $S$, and may sometimes
even use $m_i$ if this is clear in the context. Also, if $s \in S$ and
$a \in \Sigma$, then $m[s/a]$ is the vector where the value $m(s)$ is
replaced by $a$.

If $S$ is a finite set, we write $S^*$ (resp. $S^+$, $S^\omega$) for
the set of words (resp. non-empty word, infinite words) defined on
alphabet $S$.

\subsection{Concurrent games and communication graphs}

We use the model of concurrent multi-player games~\cite{BBMU15},
based on the two-player model of~\cite{AHK02}.

\begin{definition}
  A \emph{concurrent multiplayer game} is a tuple $\calG =
  \tuple{V,v_\init,\Act,\Agt,\Sigma,\Allow,\Tab, (\payoff_a)_{a \in
      \Agt}}$, where $V$ is a finite set of vertices, $v_\init \in V$
  is the initial vertex, $\Act$ is a finite set of actions, $\Agt$ is
  a finite set of players, $\Sigma$ is a finite alphabet,
  $\Allow\colon V \times \Agt \to 2^\Act\setminus\{\emptyset\}$ is a
  mapping indicating the actions available to a given player in a
  given vertex,\footnote{This condition ensures that the game is
    non-blocking.} $\Tab\colon V \times \Act^{\Agt} \to V$ associates,
  with a given vertex and a given action tuple the target vertex, for
  every $a \in \Agt$, $\payoff_a \colon V^\omega \to \bbD$ is a payoff
  function with values in a domain $\bbD \subseteq \overline{\bbR}$.
\end{definition}

An element of $\Act^\Agt$ is called a \emph{move}. Standardly
(see~\cite{AHK02} for two-player games and~\cite{BBMU15} for the
multiplayer extension), concurrent games are played as follows: from a
given vertex $v$, each player selects independently an action
(allowed by $\Allow$), which altogether form a move $m$; then, the
game proceeds to the next vertex, given by $\Tab(v,m)$; the game
continues from that new vertex.

Our setting will refine this model, in that at each round, each player
will also broadcast a message, which will be received by some of the
players. The players that can receive a message will be specified
using a communication graph that we will introduce later.  The role of
the messages will remain unclear until we commit to the definition of
a strategy.

Formally, a \emph{full history} $h$ in $\calG$ is a finite sequence
\[
v_0 \cdot (m_0,\mes_0) \cdot v_1 \cdot (m_1,\mes_1) \cdot v_2 \ldots
(m_{s-1},\mes_{s-1}) \cdot v_s \in V \cdot \left(\left(\Act^\Agt
    \times (\{0,1\}^*)^\Agt\right)\cdot V\right)^*
\]
such that for every $0 \le r < s$, for every $a \in \Agt$,
$m_r(a) \in \Allow(v_r,a)$, and $v_{r+1} = \Tab(v_r,m_r)$. For every
$0 \le r < s$, for every $a \in \Agt$, the set $\mes_r(a)$ is the
message appended to action $m_r(a)$ at step $r+1$, which will be
broadcast to some other players. For readability we will also write
$h$ as
\[
  v_0 \xrightarrow{m_0,\mes_0} v_1 \xrightarrow{m_1,\mes_1} v_2 \ldots
  \xrightarrow{m_{s-1},\mes_{s-1}} v_s
\]
We write $\vertices(h) = v_0 \cdot v_1 \cdots v_s$, and $\last(h)$ for
the last vertex of $h$ (that is, $v_s$). If $r \le s$, we also write
$h_{\ge r}$ (resp. $h_{\le r}$) for the suffix $v_r \cdot (m_r,\mes_r)
\cdot v_{r+1} \cdot (m_{r+1},\mes_{r+1}) \ldots (m_{s-1},\mes_{s-1})
\cdot v_s$ (resp. prefix $v_0 \cdot (m_0,\mes_0) \cdot v_1 \cdot
(m_1,\mes_1) \ldots (m_{r-1},\mes_{r-1}) \cdot v_r$).  We write
$\Hist_\calG(v_0)$ (or simply $\Hist(v_0)$ if $\calG$ is clear in the
context) for the set of full histories in $\calG$ that start at $v_0$.
If $h \in \Hist(v_0)$ and $h' \in \Hist(\last(h))$, then we write $h
\cdot h'$ for the obvious concatenation of histories (it then belongs
to $\Hist(v_0)$).

\medskip We add a \emph{communication (directed) graph} $G = (\Agt,E)$
to the context. The set of vertices of $G$ is the set of players, and
edges define a neighbourhood relation. An edge $(a,b) \in E$ (with $a
\ne b$) means that player $b$ can see which actions are played by
player $a$ together with the messages broadcast by player
$a$.
Later we write $a \fleche b$ whenever $(a,b) \in E$ or $a = b$, and
$\Vois(b) = \{a \in \Agt \mid a \fleche b\}$ for the so-called
\emph{neighbourhood} of $b$ (that is, the set of players about which
player $b$ has information).
If $a,b \in \Agt$, we write $\dist_G(a,b)$ for the distance in $G$
from $a$ to $b$ ($+\infty$ if there is no path from $a$ to $b$).

Let $a \in \Agt$ be a player. 
The projection of $h$ for $a$ is denoted $\pi_a(h)$ and is defined by
\begin{multline*}
  v_0 \cdot (m_0(\Vois(a)),\mes_0(\Vois(a))) \cdot v_1 \cdot
  (m_1(\Vois(a)),\mes_1(\Vois(a))) \cdot v_2 \ldots \\
  \ldots (m_{s-1}(\Vois(a)),\mes_{s-1}(\Vois(a))) \cdot v_s \in V
  \cdot \left(\left(\Act^{\Vois(a)} \times
      (\{0,1\}^*)^{\Vois(a)}\right) \cdot V \right)^*
\end{multline*}
This will be the information available to player $a$. In particular,
messages broadcast by the players are part of this information. Note
that we assume perfect recall, that is, while playing, player $a$ will
remember all her past knowledge, that is, all of $\pi_a(h)$ if $h$ has
been played so far.  We define the \emph{undistinguishability relation
  $\sim_a$} as the equivalence relation over full histories induced by
$\pi_a$: for two histories $h$ and $h'$, $h \sim_a h'$ iff $\pi_a(h) =
\pi_a(h')$. While playing, if $h \sim_a h'$, $a$ will not be able to
know whether $h$ or $h'$ has been played.  We write
$\Hist_{\calG,a}(v_0)$ for the set of histories for player $a$ (also
called $a$-histories) from $v_0$.

We extend all the above notions to infinite sequences in a
straightforward way and to the notion of \emph{full play}. We write
$\Plays_\calG(v_0)$ (or simply $\Plays(v_0)$ if $\calG$ is clear in
the context) for the set of full plays in $\calG$ that start at $v_0$.

\medskip A \emph{strategy} for a player $a \in \Agt$ from vertex $v_0$
is a mapping $\sigma_a \colon \Hist_{\calG}(v_0) \to \Act \times
\{0,1\}^*$ such that for every history $h \in \Hist_{\calG}(v_0)$,
$\sigma_a(h)[1] \in \Allow(\last(h),a)$, where the notation
$\sigma_a(h)[1]$ denotes the first component of the pair
$\sigma(h)$. The value $\sigma_a(h)[1]$ represents the action that
player $a$ will do after $h$, while $\sigma_a(h)[2]$ is the message
that she will append to her action and broadcast to all players $b$
such that $a \fleche b$.  The strategy $\sigma_a$ is said
\emph{$G$-compatible} if furthermore, for all histories $h,h' \in
\Hist(v_0)$, $h \sim_a h'$ implies $\sigma_a(h) = \sigma_a(h')$. In
that case, $\sigma_a$ can equivalently be seen as a mapping
$\Hist_{\calG,a}(v_0) \to \Act \times \{0,1\}^*$. An \emph{outcome} of
$\sigma_a$ is a(n infinite) play $\rho = v_0 \cdot (m_0,\mes_0) \cdot
v_1 \cdot (m_1,\mes_1) \ldots$ such that for every $r \ge 0$,
$\sigma_a(\rho_{\le r}) = (m_r(a),\mes_r(a))$. We write
$\out(\sigma_a,v_0)$ for the set of outcomes of $\sigma_a$ from $v_0$.

A \emph{strategy profile} is a tuple
$\sigma_\Agt=(\sigma_a)_{a \in \Agt}$, where, for every player
$a \in \Agt$, $\sigma_a$ is a strategy for player $a$. The strategy
profile is said \emph{$G$-compatible} whenever each $\sigma_a$ is
$G$-compatible. We write $\out(\sigma_\Agt,v_0)$ for the unique full
play from $v_0$, which is an outcome of all strategies part of
$\sigma_\Agt$.

When $\sigma_\Agt$ is a strategy profile and $\sigma'_d$ a player-$d$
$G$-compatible strategy, we write $\sigma_\Agt[d/\sigma'_d]$ for the
profile where player $d$ plays according to $\sigma'_d$, and each
other player $a$ ($\ne d$) plays according to $\sigma_a$. The strategy
$\sigma'_d$ is a \emph{deviation} of player $d$, or a
\emph{$d$-deviation} w.r.t. $\sigma_\Agt$.  Such a $d$-deviation is
said \emph{profitable} w.r.t. $\sigma_\Agt$ whenever
$\payoff_d \Big(\vertices(\out(\sigma_\Agt,v_0))\Big) <
\payoff_d\Big(\vertices(\out(\sigma_\Agt[d/\sigma'_d],v_0))\Big)$.

\begin{definition}
  A \emph{Nash equilibrium} from $v_0$ is a $G$-compatible strategy
  profile $\sigma_\Agt$ such that for every $d \in \Agt$, there is no
  profitable $d$-deviation w.r.t. $\sigma_\Agt$.
\end{definition}
In this definition, deviation $\sigma'_d$ needs not really to be
$G$-compatible, since the only meaningful part of $\sigma'_d$ is along
$\out(\sigma[d/\sigma'_d],v_0)$, where there are no
$\sim_d$-equivalent histories: any deviation can be made
$G$-compatible without affecting the profitability of the resulting
outcome.

\begin{remark} 
    Before pursuing our study, let us make clear what information
    players have: a player knows the full arena of the game and the
    whole communication graph; when playing the game, a player sees
    the states which are visited, and see actions of and messages from
    her neighbours (in the communication graph). When playing the
    profile of a Nash equilibrium, all players know all strategies,
    hence a player knows precisely what is expected to be the main
    outcome; in particular, when the play leaves the main outcome,
    each player knows that a deviation has occurred, even though she
    didn't see the deviator or received any message. Note that
    deviations which do not leave the main outcome may occur; in this
    case, only the neighbours of the deviator will know that such a
    deviation occurred; we will see that it is useless to take care of
    such deviations.
\end{remark}

\subsection{An example}
\label{subsec:ex}

\begin{figure}[t]
  \centering
  \begin{tikzpicture}[-latex, auto, node distance = 2 cm and 2 cm, on grid, semithick, roundnode/.style={draw, circle}, every text node part/.style={align=center}]
    
    \node[roundnode] (s0) {$v_0$};
    \node[roundnode] (s'0) [above=2 of s0] {$v'_0$};
    \node[roundnode] (s1) [left=3 of s0] {$v_1$};
    \node[roundnode] (s'1) [right=3 of s0] {$v'_1$};
    \node[roundnode] (s2) [below left = 2 and 2 of s0] {$v_2$};
    \node[roundnode] (s3) [below = 2 of s0] {$v_3$};
    \node[roundnode] (s4) [below right = 2 and 2 of s0] {$v_4$};

    \path (s0) edge[bend right]
    node[above,midway]{$\alpha^5$} (s1);
    
    \path (s0) edge[bend left]  node [above,pos=.6]
    {$\alpha^2\beta\alpha^2, \alpha^3 \beta \alpha,\alpha^4\beta$} (s'1);
    
    \path (s1) edge (s0);
    \path (s'1) edge (s0);
    
    \path (s0) edge node [midway,above,sloped]
    {$\alpha\beta\Act^3$} (s2);
    
    \path (s0) edge node [midway,sloped,above]
    {$\beta\beta\Act^3$} (s3);

    \path (s0) edge node [midway,sloped,above] {$\beta\alpha\Act^3$} (s4);
    \path (s0) edge (s'0);
    
    \path [draw,-latex'] (s3) .. controls +(-150:1.5cm) and +(150:1.5cm)
    .. (s3);
    
    \path [draw,-latex'] (s4) .. controls +(-150:1.5cm) and +(150:1.5cm)
    .. (s4);
    \path [draw,-latex'] (s2) .. controls +(-150:1.5cm) and +(150:1.5cm)
    .. (s2);
    \path [draw,-latex'] (s'0) .. controls +(-150:1.5cm) and +(150:1.5cm)
    .. (s'0);
    

    \begin{scope}[xshift=6cm,yshift=3cm,inner sep=1.5pt, >={Latex[length=1.5mm,width=2mm,angle'=60,open]}]
      \node (0) [] {$0$};
      \node (1) [below right = .9 of 0] {$1$};
      \node (2) [below = .7 of 1] {$2$};
      \node (4) [below left = .9 of 0] {$4$}; 
      \node (3) [below = .7 of 4] {$3$};
      \node (G1) [right=.8 of 1] {$G_1$};
      \path[draw,<->] (0)--(1);
      \path[draw,->] (3)--(4);
      \path[draw,->] (4)--(0);      
    \end{scope}

    \begin{scope}[xshift=6cm,yshift=1cm,inner sep=1.5pt, >={Latex[length=1.5mm,width=2mm,angle'=60,open]}]
      \node (0) [] {$0$};
      \node (1) [below right = .9 of 0] {$1$};
      \node (2) [below = .7 of 1] {$2$};
      \node (4) [below left = .9 of 0] {$4$}; 
      \node (3) [below = .7 of 4] {$3$};
      \node (G1) [right=.8 of 1] {$G_2$};
      \path[draw,<->] (0)--(1);
      \path[draw,->] (3)--(0);
      \path[draw,->] (4)--(0);      
    \end{scope}

    \begin{scope}[xshift=6cm,yshift=-1cm,inner sep=1.5pt, >={Latex[length=1.5mm,width=2mm,angle'=60,open]}]
      \node (0) [] {$0$};
      \node (1) [below right = .9 of 0] {$1$};
      \node (2) [below = .7 of 1] {$2$};
      \node (4) [below left = .9 of 0] {$4$}; 
      \node (3) [below = .7 of 4] {$3$};
      \node (G1) [right=.8 of 1] {$G_3$};
      \path[draw,<->] (0)--(1);
      \path[draw,->] (4)--(0);
    \end{scope}
  \end{tikzpicture}
  \caption{A five-player game (left) and three communication graphs
    (right); self-loops $a \fleche a$ omitted from the picture.  The
    action alphabet is $\{\alpha,\beta\}$. The transition function is
    represented as arrows from one vertex to another labeled with the
    action profile(s) allowing to go from the origin vertex to the
    destination one.  We write action profiles with length-5
    words. Convention: no label means complementary labels (e.g. one
    goes from $v_0$ to $v'_0$ using any action profile that is not in
    $\alpha^5+(\alpha\beta+\beta\beta+\beta\alpha)\Act^3+\alpha^2(\beta\alpha^2+\alpha\beta\alpha+\alpha^2\beta))$.}
 \label{gameex}
\end{figure}
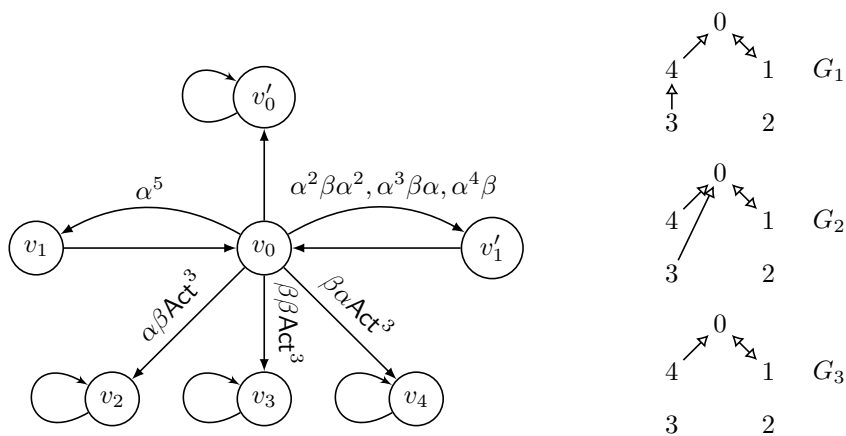

We consider the five-player game described in Figure~\ref{gameex} in
which we denote the players $\Agt = \{0,1,2,3,4\}$. The action
alphabet is $\Act=\{\alpha,\beta\}$, and the initial vertex is assumed
to be $v_0$. We suppose the payoff function vector is defined as (to
be read as the list of payoffs of the players):
\[
\payoff(\rho)= 
\left\{\begin{array}{ll}
  (0,0,1,1,1) &  \text{if}\ \rho\ \text{visits}\ v_1\ \text{infinitely
                often}\\
  (0,0,2,2,2) & \text{if}\ \rho\ \text{visits}\  v_1\ \text{finitely
                often and}\  v'_1\ \text{infinitely
                often} \\
  (0,0,0,2,2) & \text{if}\ \rho\ \text{ends
                                         up in}\ v_2\\
  (0,0,2,0,2) & \text{if}\ \rho\ \text{ends
                                         up in}\ v_3\\
  (0,0,2,2,0)  & \text{if}\ \rho\ \text{ends
                                         up in}\ v_4\\
  (0,0,3,3,3) & \text{if}\ \rho\ \text{ends
                                         up in}\ v'_0
\end{array}\right.
\]

We consider a (partial) strategy profile $\sigma$ whose main outcome
is:
\[
\rho = \big(v_0 \cdot (\alpha^5,\mes_\epsilon) \cdot v_1 \cdot
(\alpha^5,\mes_\epsilon) \big)^\omega
\]
where $\mes_\epsilon(a) = \epsilon$ for every $a \in \Agt$.  Note that
players $0$ and $1$ cannot benefit from any deviation since their
payoffs is uniformly $0$. Then notice that no one alone can deviate
from $\rho$ to $v'_0$. Now, the three players $2$, $3$ and $4$ can
alone deviate to $v'_1$ and (try to) do so infinitely often. We
examine those deviations. If players $0$ and $1$ manage to learn who
is the deviator, then, together, they can punish the deviator: if they
learn that player $2$ (resp. $3$, $4$) is the deviator, then they will
enforce vertex $v_2$ (resp. $v_3$, $v_4$). If they do not manage to
learn who is the deviator, then they will not know what to do, and
therefore, in any completion of the strategy profile, there will be
some profitable deviation for at least one of the players (hence there
will not be any Nash equilibrium whose main outcome is $\rho$).

We examine now the three communication graphs $G_1$, $G_2$ and $G_3$
depicted in Figure~\ref{gameex}. Using communication based on graph
$G_1$, if player $4$ deviates, then player $0$ will see this
immediately and will be able to communicate this fact to player $1$;
if player $3$ deviates, then player $4$ will see this immediately and
will be able to communicate this fact to player $0$, which will
transmit to player $1$; if player $2$ deviates, then no one will see
anything, hence they will deduce the identity of the deviator in all
the cases.

Using communication based on graph $G_2$, if either player $3$ or
player $4$ deviates, then player $0$ will see this immediately and
will be able to communicate this fact to player $1$ using the richness
of the communication scheme (words over $\{0,1\}$). Like before, the
identity of deviator $2$ will be guessed after a while.

Using communication based on graph $G_3$, if player $4$ deviates, then
player $0$ will see this immediately and will be able to communicate
this fact to player $1$ (as before); now, no one (except players $2$
and $3$) will be able to learn who is deviating, if player $2$ or
player $3$ deviates. 

We can conclude that there is a Nash equilibrium with graph $G_1$ or
$G_2$ whose main outcome is $\rho$, but not with graph $G_3$.

\subsection{Two-player turn-based game structures}

Two-player turn-based game structures are specific cases of the
previous model, where at each vertex, at most one player has more than
one action in her set of allowed actions. But for convenience, we will
give a simplified definition, with only objects that will be useful.

{ \emergencystretch 2cm
A two-player turn-based game structure is a tuple
$G = \tuple{S,S_\Eve,S_\Adam,s_\init,A,\Allow,\Tab}$, where
$S = S_\Eve \sqcup S_\Adam$ is a finite set of states (states in
$S_\Eve$ belong to player \Eve whereas states in $S_\Adam$ belong to
player \Adam), $s_\init \in S$ is the initial state, $A$ is a finite
alphabet, $\Allow \colon S \to 2^A \setminus \{\emptyset\}$ gives the
set of available actions, and $\Tab \colon S \times A \to S$ is the
next-state function. If $s \in S_\Eve$ (resp. $S_\Adam$), $\Allow(s)$
is the set of actions allowed to \Eve (resp. \Adam) in state $s$.}

In this context, strategies will use sequences of states.
That is, if $a$ denotes \Eve or \Adam, an $a$-strategy is a partial
function $\sigma_a: S^* \cdot S_a \to A$ such that for every
$H \in S^* \cdot S_a$ such that $\sigma_a(H)$ is defined,
$\sigma_a(H) \in \Allow(\last(H))$. Note that we do not include any
winning condition or payoff function in the tuple, hence the name
structure.

\subsection{The problems we are looking at}

We are interested in the constrained existence of a Nash
equilibrium. For simplicity, we define rectangular threshold
constraints, but could well impose more complex constraints, like
Boolean combinations of linear constraints.

\begin{problem}[Constrained existence problem]
  \emergencystretch 1.5cm Given a concurrent game $\calG = \langle
  V,v_\init,\Agt,\Act,\Sigma,\Allow,\Tab, (\payoff_a)_{a \in \Agt}
  \rangle$, a communication graph $G$ for $\Agt$, a predicate $P$ over
  $\mathbb{R}^{|\Agt|}$, can we decide whether there exists a Nash
  equilibrium $\sigma_\Agt$ from $v_\init$ such that
  $\payoff(\vertices(\out(\sigma_\Agt,v_\init))) \in P$?  If so,
  compute one.
  If the predicate $P$ is trivial, we simply speak of the existence
  problem.
\end{problem}

The case where the communication graph has no edge was studied in
depth in~\cite{BBMU15}, with a generic two-player construction called
the suspect construction, allowing to decide the constrained existence
problem for many kinds of payoff functions.  The case where the
communication graph is a clique was the subject of the
work~\cite{brenguier16}. The general case of a communication graph has
not been investigated so far, but induces interesting developments. In
the next section, we show that we can restrict the search of Nash
equilibria to the search of so-called normed strategy profiles, where
the communication via messages follows a very simple pattern. We also
argue that deviations which do not impact the visited vertices should
not be considered in the analysis. Given those reductions, we then
propose the construction of a two-player game, which will track those
normed profiles. This construction is inspired by the suspect-game
construction of~\cite{BBMU15} and of the epistemic game
of~\cite{bouyer18}.

\section{Reduction to profiles following a simple communication mechanism}
\label{sec:reducs}

We fix a concurrent game
$\calG = \tuple{V,v_\init,\Act,\Agt,\Sigma,\Allow,\Tab, (\payoff_a)_{a
    \in \Agt}}$ and a communication graph $G$. We assume that
$v_\init = v_0$. We will reduce the search for Nash
equilibria to the search for strategy profiles with a very specific
shape. In particular, we will show that the richness of the
communication offered by the setting is somehow useless, and that a
very simple communication pattern will be sufficient for 
characterizing Nash equilibria.

In the following, we write $\mes_\epsilon$ for the vector assigning
the empty word $\epsilon$ to every player $a \in \Agt$. Furthermore,
for every $d \in \Agt$, we pick some word $\id_d \in \{0,1\}^+$ which
are all distinct (and different from $\epsilon$).

\smallskip We first define restrictions for deviations.  Let
$\sigma_\Agt$ be a strategy profile. A player-$d$ deviation
$\sigma'_d$ is said \emph{immediately visible} whenever, writing $h$
for the longest common prefix of $\out(\sigma_\Agt,v_0)$ and
$\out(\sigma_\Agt[d/\sigma'_d],v_0)$,
$\Tab(\last(h),m) \ne \Tab(\last(h),m')$, where
$m = \sigma_\Agt(h)[1]$ and
$m' = \big(\sigma_\Agt[d/\sigma'_d](h)\big)[1]$ are the next moves
according to $\sigma_\Agt$ and $\sigma_\Agt[d/\sigma'_d]$. That is, at
the first position where player $d$ changes her strategy, it becomes
public information that a deviation has occurred (even though some
players know who deviated~---all the players $a$ with
$d \fleche a$---, and some other don't know).  It is furthermore
called \emph{honest} whenever for every
$h' \in \out(\sigma_\Agt[d/\sigma'_d],v_0)$ such that $h$ is a
(non-strict) prefix of $h'$, $\sigma'_d(h')[2] = \id_d$. Somehow,
player $d$ admits she deviated, and does so immediately and forever.

\smallskip The simple communication mechanism that we will design
consists in reporting the deviator (role of the direct neighbours of
the deviator), and propagating this information along the
communication graph (for all the other players). Formally, let
$\sigma_\Agt$ be a strategy profile, and let $\rho$ be its main
outcome. The profile $\sigma_\Agt$ will be said \emph{normed} whenever
the following conditions hold:
\begin{enumerate}
\item for every
  $h \in \out(\sigma_\Agt) \cup \bigcup_{d \in \Agt,\ \sigma'_d}
  \out(\sigma_\Agt[d/\sigma'_d],v_0)$, if $\vertices(h)$ is a prefix
  of $\vertices(\rho)$, then for every $a \in \Agt$,
  $\sigma_a(h)[2] = \epsilon$;
\item for every $d \in \Agt$, for every $d$-strategy $\sigma'_d$,
  if $h \cdot (m,\mes) \cdot v \in \out(\sigma_\Agt[d/\sigma'_d],v_0)$
  is the first step out of $\vertices(\rho)$, then for every $d
  \fleche a$,
  $\sigma_a(h \cdot (m,\mes) \cdot v)[2] = \id_d$; 
\item for every $d \in \Agt$, for every $d$-strategy $\sigma'_d$, if
  $h \cdot (m,\mes) \cdot v \in \out(\sigma_\Agt[d/\sigma'_d],v_0)$
  has left the main outcome for more than one step, then for every $a
  \in \Agt$, $\sigma_a(h \cdot (m,\mes) \cdot v)[2] = \epsilon$ if for
  all $b \fleche a$, $\mes(b) = \epsilon$ and $\sigma_a(h \cdot
  (m,\mes) \cdot v)[2] = \id_d$ if there is $b \fleche a$ such that
  $\mes(b) = \id_d$; note that this is well defined since at most one
  id can be transmitted.
\end{enumerate}
The first condition says that, as long as a deviation is not visible,
then no message needs to be sent;
the second condition says that as soon as a deviation becomes visible,
then messages denouncing the deviator should be sent by ``those who
know'', that is, the (immediate) neighbours of the deviator; the third
condition says that the name (actually, the id) of the deviator should
propagate according to the communication graph in an epidemic way.

Note that the profiles discussed in Section~\ref{subsec:ex} 
were actually normed.


\begin{restatable}{theorem}{coro}
  \label{coro}
  The existence of a Nash equilibrium $\sigma_\Agt$ with payoff $p$ is
  equivalent to the existence of a normed strategy profile
  $\sigma'_\Agt$ with payoff $p$, which is resistant to immediately
  visible and honest single-player deviations.
\end{restatable}

The proof of this theorem is rather technical, hence postponed to
Appendix~\ref{app:reducs}, page~\pageref{app:reducs}. We only
give some intuition here. First, we explain why being resistant to
immediately visible and honest deviations is enough. Notice that as
long as the sequence of vertices follows the main outcome, then one
can simply ignore the deviation and act only when the deviation
becomes visible, in a way as if the deviator had started deviating
only at this moment. This will be enough to punish the deviator. The
``honest'' part comes from the fact that one should simply ignore the
messages sent by the deviator as it can be only in her interest to not
ignore them (if it was not, then why would she send any message at
all?).

Second, we show why no one should communicate as long as the sequence
of vertices follows the main outcome. The reason is that if no one has
deviated then any message is essentially useless, and if a deviation
has happened, as explained earlier it can just be ignored as long as
it has not become visible.

Finally, we demonstrate why the richness of the communication
mechanism is in a way useless. Intuitively, one can understand that
the only factors that should matter when playing are the sequence of
the vertices that have been visited (because payoff functions only
take into account the visited vertices) and the identity of the
deviator. Thus the messages should only be used so that players can
know of the identity of the deviator in the fastest possible way, and
we show that nothing is faster than a sort of epidemic mechanism where
one simply broadcasts the identity of the deviator whenever one
received the information.




\section{The epistemic game abstraction}
\label{sec:epistemic}

We fix a concurrent game
$\calG = \tuple{V,v_\init,\Act,\Agt,\Sigma,\Allow,\Tab, (\payoff_a)_{a
    \in \Agt}}$ for the rest of the section, and $G$ be a
communication graph for $\Agt$. We will implement an epistemic
abstraction, which will track normed strategy profiles, and check that
there is no profitable immediately visible and honest single-player
deviations.

\subsection{Description of the epistemic game}

A \emph{situation} is a triple $(d,I,K)$ in $\Agt \times 2^\Agt \times
\Big(2^{\Agt}\Big)^\Agt$, which consists of a deviator $d \in \Agt$, a
list of players $I$ having received the information that $d$ is the
deviator, and a knowledge function $K$ that associates to every player
$a$ a list of suspects $K(a)$; in particular, it should be the case
that $d\in I$ and for every $a \in I$, $K(a) = \{d\}$. We write
$\Ssit$ for the set of situations.

\medskip The epistemic game $\calE_\calG^G$ of $\calG$ and $G$ is
defined as a two-player game structure
$\tuple{S,S_\Eve,S_\Adam,s_\init,\Sigma',\Allow',\Tab'}$.  We describe
the states and the transitions leaving those states; in particular,
components $\Sigma'$, $\Allow'$, $\Tab'$ of the above tuple will only
be implicitely defined.

\Eve's states $S_\Eve$ consist of elements of $V \times 2^\Ssit$ such
that if $(v,X)$ is a state then for all $a \in \Agt$ the set
$\{(d,I,K)\in X \mid d=a\}$ is either a singleton or empty (there is
at most one situation associated with a given player $a$). We write
$\dev(X)$ the set $\{d \in \Agt \mid \exists (d,I,K) \in X\}$ of
agents which are a deviator in one situation of $X$. If $d \in
\dev(X)$, we write $(d,I^X_d,K^X_d)$ for the unique triple belonging
to $X$ having deviator $d$. Hence, $X = \{(d,I^X_d,K^X_d) \mid d \in
\dev(X)\}$. Intuitively, an \Eve's state $(v,X)$ will correspond to a
situation where the game has proceeded to vertex $v$, but, if $\dev(X)
\ne \emptyset$, several players may have deviated. Each player $d \in
\dev(X)$ may be responsible for the deviation; some people will have
received a message denouncing $d$ (those are in the set $I^X_d$), and
some will deduce things from what they observe (this is given by
$K^X_d$). Note that the (un)distinguishability relation of a player
$a$ will be deduced from $X$: if $d$ deviated and $a \in I^X_d$, then
$a$ will know $d$ deviated; if $a$ is neither in $I^X_d$ nor in
$I^X_{d'}$, then $a$ will not be able to know whether $d$ or $d'$
deviated (as we will prove later, in Lemma~\ref{lma}).

\smallskip First let us consider the case where $X=\emptyset$, which
is to be understood as the case where no deviation has arisen yet. In
state $(v,\emptyset)$, \Eve's actions are moves in $\calG$ enabled in
$v$. When she plays move $m \in \Act^\Agt$, the game progresses to
\Adam's state $((v,\emptyset),m) \in S_\Adam$ where \Adam's actions
are vertices $v' \in V$ such that there exists a player $d \in \Agt$
and an action $\delta \in \Act$ such that
$\Tab(v,(m[d/\delta]))=v'$. When \Adam plays $v'$, either
$v'=\Tab(v,m)$ and the game progresses to \Eve's state
$(v',\emptyset)$ or $v'\neq \Tab(v,m)$ and the game progresses to
\Eve's state $(v',X')$ where:
\begin{itemize}
\item $d \in \dev(X')$ if and only if there is $\delta \in \Act$ such
  that $\Tab(v,(m[d/\delta]))=v'$. It means that given the next state
  $v'$, $d$ is a possible deviator;
\item if $d \in \dev(X')$, then:
  \begin{itemize}
  \item $I^{X'}_d = \{a \in \Agt \mid d\fleche a\}$;
  \item for every $a \in I_d^{X'}$, $K_d^{X'}(a) = \{d\}$;
  \item for every $a \notin I^{X'}_d$, $K^{X'}_d(a) = \{b\in \Agt \mid
    \exists \beta \in \Act\ \text{s.t.}\ \Tab(v,(m[b/\beta])) = v'\}
    \setminus \{b\in \Agt \mid b \fleche a\}$.
    Those are all the players that can be suspected by $a$, given the
    vertex $v'$, and the absence of messages so far.
  \end{itemize}
\end{itemize}
We write $X' = \update((v,\emptyset),m,v')$. Note that $X' =
\emptyset$ whenever (and only when) $\Tab(v,m) = v'$.

\smallskip In a state $(v,X) \in S_\Eve$ where $X\neq \emptyset$,
\Eve's actions consist of functions from $\dev(X)$ to $\Act^\Agt$ that
are compatible with players' knowledge, that is: $f: \dev(X) \to
\Act^\Agt$ is an action enabled in $(v,X)$ if and only if (i) for all
$d \in \dev(X)$, for each $a \in \Agt$, $f(d)(a) \in \Allow(v,a)$,
(ii) for all $d,d' \in \dev(X)$, for all $a \in \Agt$, if $a \notin
I^X_d \cup I^X_{d'}$ and $K^X_d(a)=K^X_{d'}(a)$ then $f(d)(a) =
f(d')(a)$;\footnote{Note in particular that ``$K^X_d(a)$ singleton''
  does not imply $a \in I_d^X$, those are two distinguishable
  situations: the message with the identity of the deviator may not
  have been received in the first case, while it has been received in
  the second case.} that is, if a player has not received any message
so far but has the same knowledge about the possible deviators in two
situations, then \Eve's suggestion for that player's action must be
the same in both situations. When \Eve plays action $f$ in $(v,X)$,
the next state is $((v,X),f) \in S_\Adam$, where \Adam's actions
correspond to states of the game that are compatible with $(v,X)$ and
$f$, that is states $v'$ such that there exists $d \in \dev(X)$ and
$\delta\in \Act$ such that $\Tab(v,f(d)[d/\delta])=v'$.

When \Adam chooses the action $v'$ in $((v,X),f)$, the game progresses
to \Eve's state $(v',X')$, where:
\begin{itemize}
\item $d \in \dev(X')$ if and only if $d\in \dev(X)$ and there exists
  $\delta \in \Act$ such that $\Tab(v,f(d)[d/\delta])=v'$. It
  corresponds to a case where $d$ was already a possible deviator and
  can continue deviating so that the game goes to $v'$;
\item if $d \in \dev(X')$, then:
  \begin{itemize}
  \item $I_d^{X'} = I_d^X \cup \{a \in \Agt \mid \exists b \in I_d^X\
    \text{s.t.}\ b \fleche a\}$. New players receive a message with
    the deviator id;
  \item for every $a \in I_d^{X'}$, $K_d^{X'}(a) = \{d\}$;
  \item for every $a \notin I_d^{X'}$, $K_d^{X'}(a) = \{b \in K_d^X(a)
    \mid \exists \beta \in \Act\ \text{s.t.}\
    \Tab(v,f(b)[b/\beta])=v'\} \setminus \{c \in \Agt \mid
    \dist_G(c,a) \le \max \{\dist_G(c,c') \mid c' \in I_c^X\} +1\}$.
    Those are the players that could have deviated but for which
    player $a$ would not have received the signal yet.
  \end{itemize}
\end{itemize}
We write $X' = \update((v,X),f,v')$. Note that $X' \ne \emptyset$ and
that $\dev(X') \subseteq \dev(X)$.

We let $R = (v_0,X_0) \cdot ((v_0,X_0),f_0) \cdot (v_1,X_1) \dots$ be
an infinite play from $(v_0,X_0) = (v_\init,\emptyset)$.  We write
$\visited(R)$ for $ v_0 v_1 \dots \in V^\omega$ the sequence of
vertices visited along $R$.  We also define $\dev(R) = \emptyset$ if
$X_r = \emptyset$ for every $r$, and $\dev(R) =\lim_{r \to +\infty}
\dev(X_r)$ otherwise. This is the set of possible deviators along $R$.

\subsubsection{Winning condition of \Eve.}

A zero-sum game will be played on the game structure $\calE_\calG^G$,
and the winning condition of \Eve will be given on the branching
structure of the set of outcomes of a strategy for \Eve, and not
individually on each outcome, as standardly in two-player zero-sum
games. We write $s_\init = (v_\init,\emptyset)$ for the initial
state. Let $p = (p_a)_{a \in \Agt} \in \bbR^\Agt$, and $\zeta$ be a
strategy for \Eve in $\calE_\calG^G$; it is said winning for $p$ from
$s_\init$ whenever $\payoff(\visited(R)) = p$, where $R$ is the unique
outcome of $\zeta$ from $s_\init$ where \Adam complies to \Eve's
suggestions, and for every other outcome $R'$ of $\zeta$, for every
$d \in \dev(R')$, $\payoff_d(\visited(R')) \le p_d$.

\subsection{An example}

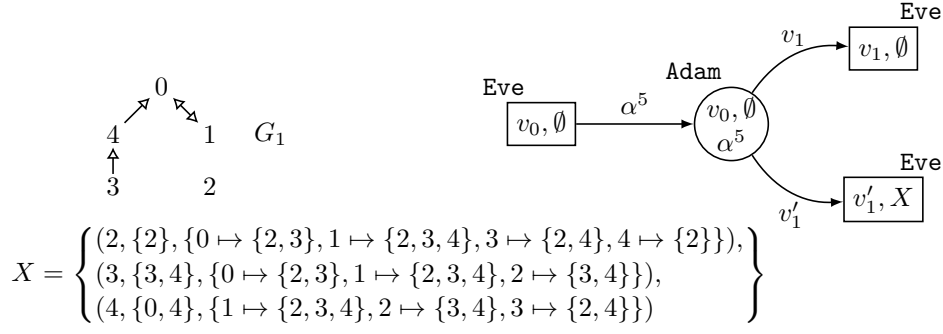
\begin{figure}[t]
\centering
 \begin{tikzpicture}[-latex, auto, node distance = 2 cm and 2 cm, on
   grid, semithick, rectnode/.style={draw, rectangle},
   rondnode/.style={draw, circle, inner sep=0pt}, every text node part/.style={align=center}]
  \node[rectnode] (init) {$v_0,\emptyset$};
  \node[rondnode] (Adam) [right = 2.5cm of init] {$v_0,\emptyset$ \\ $\alpha^5$};
  \node[rectnode] (complyEve) [above right = 1cm and 2cm of Adam] {$v_1,\emptyset$};
  \node[rectnode] (notComplyEve) [below right = 1cm and 2cm of Adam]
  {$v'_1,X$};
  \node (initStamp) [above left = 0.7 cm of init] {\Eve};
  \node (AdamStamp) [above left = 0.7 cm and .5 cm of Adam] {\Adam};
  \node (complyEveStamp) [above right = 0.7 cm of  complyEve] {\Eve};
  \node (notComplyEveStamp) [above right = 0.7 cm of notComplyEve] {\Eve};
  
  \path (init) edge node [above,midway] {$\alpha^5$} (Adam) ;
  \path (Adam) edge[bend left] node [above,midway] {$v_1$} (complyEve);
  \path (Adam) edge[bend right] node [below,midway] {$v'_1$}
  (notComplyEve);

    \begin{scope}[xshift=-5cm,yshift=.5cm,inner sep=1.5pt,
      >={Latex[length=1.5mm,width=2mm,angle'=60, open]}]
      \node (0) [] {$0$};
      \node (1) [below right = .9 of 0] {$1$};
      \node (2) [below = .7 of 1] {$2$};
      \node (4) [below left = .9 of 0] {$4$}; 
      \node (3) [below = .7 of 4] {$3$};
      \node (G1) [right=.8 of 1] {$G_1$};
      \path[draw,<->] (0)--(1);
      \path[draw, ->] (3)--(4);
      \path[draw,->] (4)--(0);      
    \end{scope}

    \begin{scope}[yshift=-2cm,xshift=-2cm]
      \node (0,0) {$X = \left\{\begin{array}{@{}l@{}} (2,\{2\},\{0\mapsto\{2,3\},1 \mapsto
     \{2,3,4\},3\mapsto \{2,4\},4\mapsto\{2\}\}), \\
     (3,\{3,4\},\{0\mapsto \{2,3\},1\mapsto\{2,3,4\},2\mapsto
                                 \{3,4\}\}), \\
                                 (4,\{0,4\},\{1 \mapsto
                                 \{2,3,4\},2\mapsto\{3,4\},3 \mapsto
                                 \{2,4\}\}) \end{array}\right\}$};
    \end{scope}

 \end{tikzpicture}
 \caption{Part of the epistemic game corresponding to the game
   described in Figure \ref{gameex} (with graph $G_1$).  This does not
   represent the whole epistemic game and a lot of actions accessible
   in the states we show here are not written.  In situations
   $(d,I,K)$ we describe $K$ by the list of its values for players
   $a\notin I$, as for all $a$ in $I$ we have $K(a)=\{d\}$ by
   definition.}
 \label{epistemicex}
\end{figure}

In Figure~\ref{epistemicex} we present a part of the epistemic game
corresponding to the game we described in Figure \ref{gameex} with
graph $G_1$. In state $(v_0,\emptyset)$, \Eve can play the action
profile $\alpha^5$ and make the game go to
$((v_0,\emptyset),\alpha^5)$ where \Adam can either play $v_1 =
\Tab(v_0,\alpha^5)$ (we say that \Adam complies with \Eve) or choose a
different state accessible from $v_0$ and an action profile that
consists in a single-player deviation from $\alpha$, for instance
$v'_1=\Tab(v_0,\alpha^2\beta\alpha^2) = \Tab(v_0,\alpha^3\beta\alpha)
=\Tab(v_0,\alpha^4\beta) $. If \Adam chooses $v'_1$, then three
players are possible deviators: $2$, $3$ and $4$. We write $X$ for the
corresponding set of situations, and we already know that $\dev(X) =
\{2,3,4\}$ .
\begin{itemize}
\item If player $2$ is the deviator, then no one (except himself)
  directly receives this information. Player $0$ knows that player $4$
  did not deviate (since $4 \fleche 0$ in $G_1$), hence
  $K_2^X(0)=\{2,3\}$; Player $1$ has no information hence $K_2^X(1) =
  \{2,3,4\}$; Player $3$ knows that he is not the deviator but cannot
  know more, hence $K_2^X(3)=\{2,4\}$; Finally, player $4$ can deduce
  many things: he knows he is not the deviator, and he saw that player
  $3$ is not the deviator (since $3 \fleche 4$ in $G_1$), hence
  $K_2^X(4) = \{2\}$.
\item If player $3$ is the deviator, then both players $3$ and $4$ get
  the information, hence $I^X_3=\{3,4\}$. Other players can guess some
  things, for instance player $0$ sees that player $4$ cannot be the
  deviator, this is why $K^X_3(0) = \{2,3\}$. Etc.
\item The reasoning for player $4$ is similar.
\end{itemize}
In the situation we have just described, when the game will proceed to
$v'_1$, then either player $0$ knows that player $4$ has deviated, or
he knows that player $4$ didn't deviate but he suspects both $2$ and
$3$. On the other hand, player $4$ will precisely know who
deviated. And player $3$ knows whether he deviated or not, but if he
didn't, then he cannot know whether it was player $2$ or player $4$
who deviated. This knowledge is stored in situation $X$ we have
described, and which is fully given in Figure~\ref{epistemicex}.

\smallskip Let us now illustrate how actions of \Eve are defined
  in states with a non-empty set of situations. Assume we are in
  \Eve's state $(v_0,X)$, with $X$ as previously defined. From that
  state, an action for \Eve is a mapping
  $f : \{2,3,4\} \to \Act^\Agt$ such that:
  \[
    f(2)(0) = f(3)(0) \quad f(2)(1) = f(3)(1) = f(4)(1) \quad f(3)(2) =
    f(4)(2) \quad f(2)(3) = f(4)(3)
  \]
  The intuition behind these constraints is the following: Player $0$
  knows whether Player $4$ deviated or not, but in the case she did not
  cannot know whether Player $2$ or Player $3$ deviated; Player $1$ does 
  not know who deviated, hence should play the same action in the three 
  cases (that she cannot distinguish); Player $2$ does only know whether 
  she deviated hence in the case she did not cannot know whether Player $3$
  or Player $4$ deviated; the case for Player $3$ is similar; 
  finally Player $4$ knows for sure who deviated: she saw if Player $3$ 
  deviated and knows whether she herself deviated, thus can distinguish
  between the three cases.

\subsection{Correctness of the epistemic game construction}

When constructing the epistemic game, we mentioned that \Eve's states
will allow to properly define the undistinguishability relation for
all the players. Towards that goal, we show by an immediate induction
the following result:

\begin{restatable}{lemma}{lma}\label{lma}
  If $(v,X)$ is an \Eve's state reachable from some $(v_0,\emptyset)$ in
  $\calE_\calG^G$, then for all $d \in \dev(X)$:
  \begin{itemize}
  \item for all $a \in I^X_d$, $K^X_d(a) = \{d\}$;
  \item for all $a \notin I^X_d$,
    $K^X_d(a) = \dev(X) \setminus \{d' \in\dev(X) \mid a \in I^X_{d'}\}$.
  \end{itemize}
  In particular, for all $d,d' \in \dev(X)$, for all
  $a \notin I^X_d \cup I^X_{d'}$, $K^X_d(a) = K^X_{d'}(a)$.
\end{restatable}
So, either a player $a$ will have received from a neighbour the
identity of the deviator, or she will not have received any deviator
identity yet, and she will have a set of suspected deviators that she
will not be able to distinguish.

This allows to deduce the following correspondence between $\calG$ and
$\calE_\calG^G$:

\begin{proposition}
  There is a winning strategy for \Eve in $\calE_\calG^G$ for payoff
  $p$ if and only if there is a normed strategy profile in $\calG$,
  whose main outcome has payoff $p$ and which is resistant to
  single-player immediately visible and honest deviations.
\end{proposition}

\smallskip The proof of correctness of the epistemic game then goes
through the following steps, which are detailed in
Appendix~\ref{sec:correctness}, page~\pageref{sec:correctness}. First,
given an $\Eve$'s strategy $\zeta$, we build a function $E_\zeta$
associating with $a$-histories (for every $a \in \Agt$) in the
original game $\Eve$'s histories in the epistemic game such that
$\Eve$ plays according to $\zeta$ along $E_\zeta$.

Then we use this function to create a strategy profile $\Omega(\zeta)$
for the original game where the action prescribed by this profile to
player $a$ after history $h$ corresponds in some sense to
$\zeta(E_\zeta(h))(d)(a)$, where $d$ is a suspected deviator according
to player~$a$. This works because, thanks to Lemma~\ref{lma}, we know
that either player $a$ knows who the deviator is, or player $a$ has a
subset of suspect deviators and $\Eve$'s suggestion for $a$ (by
construction of $\calE_\calG^G$) is the same for all those possible
deviators.

Finally we prove that if $\zeta$ is a winning strategy for $\Eve$ then
$\Omega(\zeta)$ is both normed and resistant to single-player
immediately visible and honest deviations in $\calG$.

\smallskip To prove the converse proposition we build a function
$\Lambda$ associating with $\Eve$'s histories in the epistemic game
families of single-player histories in the original game. We then use
this correspondence to build a function $\Upsilon$ associating with
normed strategy profiles $\Eve$'s strategies in a natural way.

Finally we prove that if $\sigma$ is normed and resistant to
single-player immediately visible honest deviations, then
$\Upsilon(\sigma)$ is a winning strategy for $\Eve$.


\medskip Gathering results of Theorem~\ref{coro} and of this
proposition, we get the following theorem:

\begin{restatable}{theorem}{main}
  \label{theo:main}
  There is a Nash equilibrium with payoff $p$ in $\calG$ if and only
  if there is a winning strategy for \Eve in $\calE_\calG^G$ for
  payoff $p$.
\end{restatable}

\begin{remark}
  Note that all the results are constructive, hence if one can
  synthesize a winning strategy for \Eve in $\calE_\calG^G$, then one
  can synthesize a correponding Nash equilibrium in $\calG$.
\end{remark}

\section{Complexity analysis}
\label{sec:complexity}

We borrow all notations of previous sections. A rough analysis of the
size of the epistemic game $\calE_\calG^G$ gives an exponential
bound. We will give a more precise bound, pinpointing the part with an
exponential blowup.
We write $\diam(G)$ for the diameter of $G$, that is $\diam(G) = \max
\{\dist_G(a,b) \mid \dist_G(a,b) <+\infty\}$.



\begin{restatable}{lemma}{size}
  Assuming that $\Tab$ is given explicitely in $\calG$, the number of
  states in the reachable part of $\calE_\calG^G$ from $s_\init =
  (v_\init,\emptyset)$ is bounded by
  \[
  |S_\Eve| \le |V| + |V| \cdot |\Tab|^2 \cdot (\diam(G)+2)\qquad
  \text{and} \qquad |S_\Adam| \le |S_\Eve| \cdot |\Act|^{|\Agt|^2}
  \]
  The number of edges is bounded by $|S_\Adam| + |S_\Adam| \cdot
  |S_\Eve|$.

  If $|\Agt|$ is assumed to be a constant of the problem, then the
  size of $\calE_\calG^G$ is polynomial in the size of $\calG$.
\end{restatable}


We will not detail algorithmics issues, but the winning condition of
\Eve in $\calE_\calG^G$ is very similar to the winning condition of
\Eve in the suspect-game construction of~\cite{BBMU15} (for Boolean or
ordered objectives), or in the deviator-game construction
of~\cite{brenguier16} (for mean-payoff), or in a closer context to the
epistemic-game construction of~\cite{bouyer18}. Hence, when the size
of the epistemic game is polynomial, rather efficient algorithms can
be designed to compute Nash equilibria. For instance, in a setting
where the size of $\calE_\calG^G$ is polynomial, using a bottom-up
labelling algorithm similar to that of~\cite[Sect. 4.3]{bouyer17}, one
obtains a polynomial space algorithm for deciding the (constrained)
existence of a Nash equilibrium when payoffs are Boolean payoffs
corresponding to parity conditions.




\section{Conclusion}
\label{sec:conclusion}

In this paper, we have studied multiplayer infinite-duration games
over graphs, and focused on games where players can communicate with
neighbours, given by a directed graph. We have shown that a very
simple communication mechanism was sufficient to describe Nash
equilibria. This mechanism is sort of epidemic, in that if a player
deviates, then his neighbours will see it and transmit the information
to their own neighbours; the information then propagates along the
communication graph. This framework encompasses two standard existing
frameworks, one where the actions are invisible (represented with a
graph with no edges), and one where all actions are visible
(represented by the complete graph). We know from previous works that
in both frameworks, one can compute Nash equilibria for many kinds of
payoff functions. In this paper, we also show that we can compute Nash
equilibria in this generalized framework, by providing a reduction to
a two-player game, the so-called epistemic game construction. Winning
condition in this two-player game is very similar to winning
conditions encountered in the past, yielding algorithmic solution to
the computation of Nash equilibria. We have also analyzed the size of
the abstraction, which is polynomial when the number of players is
considered as a constant of the problem.

The current framework assumes messages can be appended to actions by
players, allowing a rich communication between players. The original
framework of~\cite{RT98} did not allow additional messages, but did
encode identities of deviators by sequences of actions. This was
possible in~\cite{RT98} since games were repeated matrix games, but it
is harder to see how we could extend this approach and how we could
encode identities of players with actions, taking into account the
graph structure. For instance, due to the graph, having too long
identifiers might be prohibitive to transmit in a short delay the
identity of the deviator. Nevertheless, that could be interesting to
see if something can be done in this framework.



\appendix

\section{Proof of Section~\ref{sec:reducs}}
\label{app:reducs}

\coro*

We will divide the proof of this theorem into several parts and prove several intermediary results.

\subsection{Reduction to immediately visible and honest deviations}

\begin{restatable}{lemma}{lemmaun}
  \label{lemma1}
  Assume $\sigma_\Agt$ is a strategy profile from $v_0$ such that
  there is no profitable single-player immediately visible honest
  deviation w.r.t. $\sigma_\Agt$, then one can construct a Nash
  equilibrium $\widetilde\sigma_\Agt$ from $v_0$ such that (i)
  $\vertices(\out(\sigma_\Agt,v_0)) =
  \vertices(\out(\widetilde\sigma_\Agt,v_0))$; (ii) for every
  $h \in \Hist(v_0)$ such that $\vertices(h)$ is a prefix of
  $\vertices(\out(\widetilde\sigma_\Agt,v_0))$, for every
  $a \in \Agt$, $\widetilde\sigma_a(h)[2] = \epsilon$.
\end{restatable}
Condition (ii) in the statement means somehow that no message is
required along the main outcome and along deviations that are not
visible yet (in the sense that they follow the same sequence of
vertices as the main outcome).

\begin{proof}
  From $\sigma_\Agt$, we will build a profile $\widetilde\sigma_\Agt$
  such that:
  \begin{itemize}
  \item $\widetilde\sigma_\Agt$'s and $\sigma_\Agt$'s main outcomes
    visit the same sequence of vertices;
  \item no message is broadcast by $\widetilde\sigma_\Agt$ as long as
    the sequence of vertices of the main outcome is followed;
  \item any single-player $d$-deviation of $\widetilde\sigma_\Agt$
    will correspond to an immediately visible single-player honest
    $d$-deviation of $\sigma_\Agt$.
  \end{itemize}

  Let $\rho=\out(\sigma_\Agt,v_o)$. For every $h \in \Hist(v_0)$ such
  that $\vertices(h) = \vertices(\rho_{\le r})$ for some $r$, for
  every $a \in \Agt$, we define
  $\widetilde{\sigma}_a(h) = \big(\sigma_a(\rho_{\le
    r})[1],\epsilon\big)$.  As a consequence of this definition, the
  main outcome $\widetilde\rho$ of $\widetilde\sigma_\Agt$ visits the
  same sequence of vertices (and even follows the same moves) as the
  main outcome $\rho$ of $\sigma$; only messages which are broadcast
  may differ.  Thereafter, we write:
  \[
    \rho = v_0 \cdot (m_0,\mes_0) \cdot v_1 \cdot (m_1,\mes_1) \cdot
    v_2 \ldots (m_{r-1},\mes_{r-1}) \cdot v_r \ldots
  \]
  and, by definition of $\widetilde\sigma_\Agt$:
  \[
    \widetilde\rho = v_0 \cdot (m_0,\mes_\epsilon) \cdot v_1 \cdot
    (m_1, \mes_\epsilon) \cdot v_2 \ldots (m_{r-1}, \mes_\epsilon)
    \cdot v_r \ldots
  \]
  where $\widetilde\rho=\out(\widetilde\sigma_\Agt,v_0)$.  Note also
  that the above definition takes care of single-player deviations
  w.r.t. $\widetilde\sigma_\Agt$ which do not `leave' the main outcome
  (that is, which visit the same sequence of vertices): against such
  deviations, simply doing nothing and going on as if no one deviated
  will be enough to punish the deviator.

  \smallskip We will now extend the definition ot profile
  $\widetilde\sigma_\Agt$ to histories generated by single-player
  deviations. We will do so by induction, by considering histories
  that can be derived by single-player deviations of the
  so-far-defined profile, and which have become visible.

  Let $h$ be a history resulting from a single-player (named $d$)
  deviation w.r.t. $\widetilde\sigma_\Agt$, which has become visible.
  It can be decomposed as:
  \[
    h = v_0 \cdot (\bar{m}_0,\bar\mes_0) \cdot v_1 \cdot
    (\bar{m}_1,\bar\mes_1) \ldots v_r \cdot (\bar{m}_r,\bar\mes_r)
    \cdot h_1
  \]
  where $\textit{first}(h_1) \ne v_{r+1}$ (that is, the deviation
  becomes visible at that point!). Note however that it may be the
  case that player $d$ started deviating earlier, but this will not be
  taken into account.
  We remark that, by definition of $\widetilde\sigma_\Agt$ along the
  main outcome or for invisible deviations, it is the case that for
  every $0 \le s \le r$, $\bar{m}_s(-d) = m_s(-d)$ and
  $\bar\mes_s(-d) = \mes_\epsilon(-d)$.

  With these notations, we can set:
  \[
    \widetilde\sigma_a(h) = \sigma_a\big(h_{\le r} \cdot
    (\bar{m}_r,\mes_r^{+d}) \cdot h_1^{+d}\big)
  \]
  where $\mes_r^{+d}(d) = \id_d$,
  $\mes_r^{+d}(-d) = \mes_r(-d) = \epsilon$, and $h_1^{+d}$ is the
  same as $h_1$, but each message of player $d$ is set to $\id_d$. 
  In some sense, we tell the players to ignore any deviation as long
  as it has not become visible and then treat it as if it were an
  immediately visible and honest deviation.
  Indeed, $h_{\le r} \cdot (m'_r,\mes_r^{+d}) \cdot h_1^{+d}$ is the
  history resulting from playing  an immediately visible and honest
  deviation visiting the same vertices as $h$.
  
  We argue why this is well-defined. Pick two such deviations, for
  players $d$ and $d'$ respectively, which generate histories
  $h = h_{\le r} \cdot  (\bar{m}_r,\bar\mes_r) \cdot
  h_1$ and $h' = h'_{\le r} \cdot  (\bar{m}'_r,\bar\mes'_r) \cdot
  h'_1$ respectively.
  We assume $a \ne d,d'$. It is the case
  that $a$ cannot distinguish between $h$ and $h'$ if and only if
  $\pi_a(h) = \pi_a(h')$. By considering the players in the
  neighbourhood of $d$, it is not difficult to get that
  $\pi_a(h) = \pi_a(h')$ implies
  $\pi_a(h_{\le r} \cdot (\bar{m}_r,\mes_r^{+d}) \cdot h_1^{+d}) =
  \pi_a(h'_{\le r} \cdot (\bar{m}'_r,{\mes'_r}^{+d}) \cdot
  {h'_1}^{+d})$. Hence, this is well-defined.

  Now, let $\rho'$ be the outcome of a $d$-deviation
  w.r.t. $\widetilde\sigma_\Agt$. It can be written as
  $v_0 \cdot (\bar{m}_0,\bar\mes_0) \cdot v_1 \cdot
  (\bar{m}_1,\bar\mes_1) \ldots v_r \cdot (\bar{m}_r,\bar\mes_r) \cdot
  \rho_1$ (as before). We have that
  $\rho_{\le r} \cdot (\bar{m}_r,\mes_r^{+d}) \cdot \rho_1^{+d}$ is
  the outcome of an immediately visible and honest $d$-deviation of
  $\sigma_\Agt$ (since one can check all players but $d$ play
  according to $\sigma_\Agt$ along the play).  Hence, it cannot be
  profitable to the deviator (by hypothesis on $\sigma_\Agt$). Since
  the two sequences $\vertices(\rho')$ and
  $\vertices(\rho_{\le r} \cdot (\bar{m}_r,\mes_r^{+d}) \cdot
  \rho_1^{+d})$ coincide, we conclude that $\rho'$ is not a profitable
  deviation, and therefore that $\widetilde\sigma_\Agt$ is a Nash
  equilibrium.
\end{proof}

\subsection{A simple communication pattern is sufficient! Reduction to
  normed strategy profiles}

In this part, we will show that the richness of the communication
offered by the setting is somehow useless, in that we will show that a
very simple communication pattern will be sufficient for generating
Nash equilibria. This pattern consists in reporting the deviator (role
of the direct neighbours of the deviator), and propagating this
information (for all the other players, following the communication
graph).

We recall the notion of a normed profile.  Let $\sigma_\Agt$ be a
strategy profile. Let $\rho$ be its main outcome. The profile
$\sigma_\Agt$ is said \emph{normed} whenever the following conditions
hold:
\begin{enumerate}
\item for every
  $h \in \out(\sigma_\Agt) \cup \bigcup_{d \in \Agt,\ \sigma'_d}
  \out(\sigma_\Agt[d/\sigma'_d],v_0)$, if $\vertices(h)$ is a prefix
  of $\vertices(\rho)$, then for every $a \in \Agt$,
  $\sigma_a(h)[2] = \epsilon$;
\item for every $d \in \Agt$, for every $d$-strategy $\sigma'_d$,
  if $h \cdot (m,\mes) \cdot v \in \out(\sigma_\Agt[d/\sigma'_d],v_0)$
  is the first step out of $\vertices(\rho)$, then for every $d
  \fleche a$,
  $\sigma_a(h \cdot (m,\mes) \cdot v)[2] = \id_d$; 
\item for every $d \in \Agt$, for every $d$-strategy $\sigma'_d$, if
  $h \cdot (m,\mes) \cdot v \in \out(\sigma_\Agt[d/\sigma'_d],v_0)$
  has left the main outcome for more than one step, then for every $a
  \in \Agt$, $\sigma_a(h \cdot (m,\mes) \cdot v)[2] = \epsilon$ if for
  all $b \fleche a$, $\mes(b) = \epsilon$ and $\sigma_a(h \cdot
  (m,\mes) \cdot v)[2] = \id_d$ if there is $b \fleche a$ such that
  $\mes(b) = \id_d$; note that this is well defined since at most one
  id can be transmitted.
\end{enumerate}
The first condition says that, as long as a deviation is not visible,
then no message needs to be sent;
the second condition says that as soon as a deviation becomes visible,
then messages denouncing the deviator should be sent by ``those who
know'', that is, the (immediate) neighbours of the deviator; the third
condition says that the name (actually, the id) of the deviator should
propagate according to the communication graph.

Our first issue is that it can be the case that the deviator can
deviate in two distinct ways (two different actions or messages) but
leading to the same sequence of vertices. Those two deviations should
be treated in the same way by the other players. This might in general
not the case, hence we provide a construction to ensure it.

First, let us introduce a new equivalence relation on histories. Given
a strategy profile $\sigma_\Agt$ we say that
$h \equiv_{\sigma_\Agt} h'$ if either $h=h'$ or writing
$h=v_0\cdot(m_0,\mes_0)\cdot v_1...\cdot v_{s-1} \cdot
(m_{s-1},\mes_{s-1})\cdot v_s$ and
$h'=v_0\cdot(m'_0,\mes'_0)\cdot v'_1...\cdot v'_{s-1} \cdot
(m'_{s-1},\mes'_{s-1})\cdot v'_s$, we have:
\begin{itemize}
 \item $\vertices(h)=\vertices(h')$;
 \item there is $d\in \Agt$ such that $h$ and $h'$ are both
   $d$-deviations w.r.t. $\sigma_\Agt$, and
   $\min\{l \mid (m_l,\mes_l)\neq\sigma_\Agt(h_{\le l})\} = \min\{l
   \mid (m'_l,\mes'_l)\neq\sigma_\Agt(h'_{\le l})\}$.
\end{itemize}
Essentially both $h$ and $h'$ leave the main outcome of $\sigma_\Agt$ due to
a single-player deviation from the same player, and follow the same sequence
of vertices.

\begin{lemma}
  \label{lmaequiv}
  Let $\sigma_\Agt$ be a strategy profile which is resistant to
  immediately visible single-player honest deviations
  w.r.t. $\sigma_\Agt$. For every history $h$ we define a canonical
  representative of the equivalence class of $h$ under
  $\equiv_{\sigma_\Agt}$ which we denote by $\overline h$, with the
  constraint that whenever possible $\overline h$ will be generated by
  playing $\sigma_\Agt$ against a single-player deviation (there exists
  $d$, $\tau_d$ and $k$ such that 
  $\overline h = \out(\sigma_\Agt[d/\tau_d],v_0)_{\leq k}$).
  We define the strategy profile $\sigma'_\Agt$ inductively by
  \begin{itemize}
  \item if $h = \out(\sigma_\Agt,v_0)_{\leq k}$ for some $k$, then
    $\sigma'_\Agt(h)=\sigma_\Agt(h)$;
  \item otherwise, if $h = \out(\sigma'_\Agt[d/\tau_d],v_0)_{\leq k}$
    for some $k$, some player $d\in \Agt$ and some $d$-deviation $\tau_d$, 
    then:
    \begin{itemize}
    \item if $d \fleche a$, then $\sigma'_a(h)=\sigma_a(\overline h)$;
    \item else $\sigma'_a(h)=\sigma_a(h)$.
    \end{itemize}
  \end{itemize}
  Then, $\sigma'_\Agt$ is a strategy profile which is resistant to
  immediately visible single-player honest deviations
  w.r.t. $\sigma'_\Agt$, with the same outcome as $\sigma_\Agt$ and
  satisfies if $h\equiv_{\sigma'_\Agt} h'$ then $\sigma'_\Agt(h)=\sigma'_\Agt(h')$.
\end{lemma}

\begin{proof}
  First we argue why $\sigma'_\Agt$ is well-defined. Consider two
  histories $h$ and $h'$ generated by playing $\sigma'_\Agt$ against
  single-player deviations with deviators $d$ and $d'$ respectively,
  and a player $a$ such that $h \sim_a h'$. If $d \fleche a$ or $d'
  \fleche a$, then, since $h\sim_a h'$, we deduce that $d=d'$. Thus,
  by hypothesis, $\overline{h} = \overline{h'}$, which implies
  $\sigma'_a(h)=\sigma'_a(h') = \sigma_a(\overline{h})$. If $d
  \notfleche a$ and $d' \notfleche a$, then
  $\sigma'_a(h)=\sigma_a(h)=\sigma_a(h')=\sigma'_a(h')$.
  
  It is fairly easy to show by induction that if $h\equiv_{\sigma'_\Agt} h'$
  then $\sigma'_\Agt(h)=\sigma'_\Agt(h')$.
  
  Finally, we explain why $\sigma'_\Agt$ is resistant to immediately
  visible single-player deviations. We show by induction that for
  every $k\in\mathbb{N}$, for every deviator $d$ and every
  $d$-deviation $\tau_d$ w.r.t. $\sigma'_\Agt$ there exists a
  $d$-deviation $\tau'_d$ w.r.t. $\sigma_\Agt$ such that
  $\vertices(\out(\sigma'_\Agt[d/\tau_d],v_0)_{\leq k})=
  \vertices(\out(\sigma_\Agt[d/\tau'_d],v_0)_{\leq k})$.
  

  Hence, since $\sigma_\Agt$ is resistant to immediately visible
  single-player deviations, then so is $\sigma'_\Agt$.
\end{proof}

Let us now use this to prove the following result.

\begin{restatable}{lemma}{lemmadeux}
  \label{lemma2}
  Assume $\sigma_\Agt$ is a strategy profile which is resistant to
  immediately visible single-player honest deviations
  w.r.t. $\sigma_\Agt$ and such that, if $\vertices(h)$ is a prefix of
  $\vertices(\out(\sigma_\Agt,v_0))$, then
  $\sigma_\Agt(h)[2] = \epsilon$.  Then one can construct a normed
  strategy profile $\widehat\sigma$ which is resistant to immediately
  visible honest single-player deviations, and such that
  $\vertices(\out(\sigma_\Agt,v_0)) =
  \vertices(\out(\widehat\sigma_\Agt,v_0))$.
\end{restatable}

\begin{proof}
  Let $\sigma_\Agt$ be a strategy profile which is resistant to
  immediately visible single-player honest deviations, with no
  communication along the main outcome. Furthermore, by usage of
  Lemma~\ref{lmaequiv}, we can suppose that for all histories $h$ and
  $h'$ such that $h\equiv_{\sigma_\Agt} h'$, we have that
  $\sigma_\Agt(h)=\sigma_\Agt(h')$. We will build a profile
  $\widehat\sigma_\Agt$ which will be normed and will coincide in some
  sense to $\sigma_\Agt$.  Let $\rho$ be the main outcome of
  $\sigma_\Agt$.  We analyse immediately visible deviations from
  $\sigma_\Agt$ at some given step of $\rho$, say after prefix
  $h_0 = \rho_{\le r}$, and show that the undistinguishability
  relation of deviations from that point is then very simple.

  \begin{lemma}
    \label{lemma2}
    Let $a$ be a player, and assume $d$ and $d'$ are two players such
    that $\dist_G(d,a) > n$ and $\dist_G(d',a)>n$. Then, for each
    length-$(r+n)$ immediately visible $d$-deviation $h$
    w.r.t. $\sigma_\Agt$, after $h_0$, for each length-$(r+n)$
    immediately visible $d'$-deviation $h'$ w.r.t. $\sigma_\Agt$,
    after $h_0$, if the projection over vertices of $h$ and $h'$
    coincide, then $h \sim_a h'$.
  \end{lemma}

  \begin{proof}
    We do the proof by induction on $n$.

    \textit{Case $n=1$:} We look at one-step deviations after
    $h_0$, say $h = h_0 \cdot (m,\mes) \cdot v$ and
    $h' = h_0 \cdot (m',\mes') \cdot v$ (the final vertex $v$ is
    assumed to be the same, otherwise those two deviations will for
    sure be distinguished by every player). Then, the following holds:
    \begin{itemize}
    \item $m(-d) = \sigma_\Agt(h_0)(-d)[1]$ and
      $\mes(-d) = \mes_\epsilon(-d)$;
    \item $m'(-d') = \sigma_\Agt(h_0)(-d')[1]$ and
      $\mes'(-d') = \mes_\epsilon(-d')$.
    \end{itemize}
    The fact that $\dist_G(d,a) > n$ and $\dist_G(d',a)>n$ means that
    $d \notfleche a$ and $d' \notfleche a$.
    Hence, for every $b \fleche a$, $m(b) = \sigma_b(h_0)[1] = m'(b)$
    and $\mes(b) = \epsilon = \mes'(b)$. This implies that $\pi_a(h_0
    \cdot (m,\mes) \cdot v) = \pi_a(h_0 \cdot (m',\mes') \cdot v)$,
    which means that $h \sim_a h'$.  This concludes the case $n=1$.

    \textit{Inductive step:} assume that $\dist_G(d,a) > n+1$ and
    $\dist_G(d',a)>n+1$, and that $h \cdot (m,\mes) \cdot v$ and
    $h' \cdot (m',\mes') \cdot v$ are respectively length-$(r+n+1)$
    immediately visible $d$- (resp. $d'$-)deviations after $h_0$
    w.r.t. $\sigma_\Agt$, which project on the same sequences of
    vertices. For every $b \fleche a$, $\dist_G(d,b)>n$, hence, by
    induction hypothesis, $h \sim_b h'$. It implies that
    $\sigma_b(h) = \sigma_b(h')$, hence $m(b) = m'(b)$ and
    $\mes(b) = \mes'(b)$. We deduce that
    $h \cdot (m,\mes) \cdot v \sim_a h' \cdot (m',\mes') \cdot
    v$. This concludes the proof of the inductive step.
  \end{proof}
  We will now see that this simple undistinguishability relation can
  be defined using the message propagation mechanism of normed
  profiles: a player doesn't send any message, unless she receives the
  id of the deviator from (at least) one neighbour.

  We first define normalized versions of immediately visible
  single-player deviations, and for this we set
  $\eta(\rho_{\le r}) = \rho_{\le r}$ for every integer $r$. Let now
  $h$ be an immediately visible and honest $d$-deviation
  w.r.t. $\sigma_\Agt$, which becomes visible right after
  $\rho_{\le r}$. We define $\eta(h)$ inductively as follows:
  \begin{itemize}
  \item If $h = \rho_{\le r} \cdot (m,\mes) \cdot v$, we set
    $\eta(h) = (\rho_{\le r} \cdot (m,\mes_d) \cdot v)$ where
    $\mes_d(d) = \id_d$ and $\mes_d(-d) = \mes_\epsilon(-d)$;
  \item if $h$ already extends $\rho_{\le r}$, then $\eta(h \cdot
    (m,\mes) \cdot v) = (\eta(h) \cdot (m,\mes') \cdot v)$ where, for
    every $a \in \Agt$, $\mes'(a) = \epsilon$ if for every $b \fleche
    a$, $\bar\mes(b) = \epsilon$ and $\mes'(a) = \id_d$ if there is
    some $b \fleche a$ such that $\bar\mes(b) = \id_d$, where
    $\bar\mes$ is the last message along $\eta(h)$.
  \end{itemize}
  Somehow, $\eta$ propagates properly messages after a deviation has
  happened.

  Now define profile $\widehat\sigma_\Agt$ as follows: for every $r$,
  set $\widehat\sigma_\Agt(\rho_{\le r}) = \sigma_\Agt(\rho_{\le
    r})$. In particular, the main outcome of $\widehat\sigma_\Agt$ is
  $\rho$.  Let now $h$ be an immediately visible and honest
  $d$-deviation w.r.t. $\sigma_\Agt$, after $\rho_{\le r}$; we set for
  every $a \in \Agt$:
  \begin{itemize}
  \item if $|h| \ge \dist_G(d,a)-r$, we set
    $\widehat\sigma_a(\eta(h)) = (\sigma_a(h)[1],\id_d)$;
  \item if $|h| < \dist_G(d,a)-r$, we set
    $\widehat\sigma_a(\eta(h)) = (\sigma_a(h)[1],\epsilon)$.
  \end{itemize}
  We argue why $\widehat\sigma_\Agt$ is defined everywhere it should
  be defined. For that we show that every immediately visible
  single-player deviation of $\widehat\sigma_\Agt$ is the image by
  $\eta$ of some immediately visible single-player deviation of
  $\sigma_\Agt$. This can easily be done by induction on the length of
  the deviation.

  We then notice that, by construction, $\widehat\sigma_\Agt$
  propagates messages properly.  We finally argue below why
  $\widehat\sigma_\Agt$ is well-defined.

  Assume that $h$ and $h'$ are two histories generated by immediately visible $d$-
  (resp. $d'$-)deviations w.r.t. $\widehat\sigma_\Agt$ such that
  $\eta(h) \sim_a \eta(h')$ but such that $h \not\sim_a h'$. From
  Lemma~\ref{lemma2}, we deduce that $\dist_G(d,a)<r+n$ or
  $\dist_G(d',a)<r+n$ where $h$ and $h'$ become visible at step $r+1$
  and are of length $r+n$ (we know they are of the same length, become
  visible on the same step and that $\vertices(h)=\vertices(h')$ from
  $\eta(h)\sim_a \eta(h')$).  This means in turns thanks to the
  construction of $\widehat\sigma_\Agt$ that player $a$ received
  message $\id_d$ at some step along $\eta(h)$ or signal $\id_{d'}$ at
  some step along $\eta(h')$. Hence we have $d=d'$ and player $a$
  received message $\id_d$ at the same step along both $\eta(h)$ and
  $\eta(h')$ since otherwise, she would be able to tell the difference
  between them. Thus, $h$ and $h'$ are generated by playing two $d$-deviations
  visiting the same vertices and deviating from the main outcome at the same time. 
  Hence, we have that $h \equiv_{\sigma_\Agt} h'$,
  which means from our hypothesis on $\sigma_\Agt$ that
  $\sigma_\Agt(h)=\sigma_\Agt(h')$, thus $\widehat\sigma_\Agt$ is
  well-defined.
  
  Now we prove that $\widehat\sigma_\Agt$ is resistant to immediately
  visible single-player deviations.  Toward a contradiction, assume
  there is a profitable immediately visible $d$-deviation. As noticed
  above, this deviation is the image by $\eta$ of some immediately
  visible $d$-deviation w.r.t. $\sigma_\Agt$. Since $\eta$ preserves
  the sequences of vertices, this deviation is profitable as
  well. This is not possible, by assumption on $\sigma_\Agt$, hence
  $\widehat\sigma_\Agt$ is resistant to immediately visible
  single-player deviations.
\end{proof}

\subsection{Proof of Theorem~\ref{coro}}

\coro*

\begin{proof}
  Assume $\sigma_\Agt$ is a Nash equilibrium with payoff $p$.  It is
  in particular resistant to immediately visible single-player honest
  single-player deviations. Applying Lemma~\ref{lemma1}, $\sigma_\Agt$
  can be turned to a profile $\widetilde\sigma_\Agt$ satisfying the
  hypotheses of Lemma~\ref{lemma2}, and having payoff $p$. Applying
  Lemma~\ref{lemma2}, there is another profile $\widehat\sigma_\Agt$,
  which is normed and resistant to immediately visible single-player
  deviations, and such that its main outcome visits the same sequence
  of vertices as the main outcome of $\sigma_\Agt$. In particular, it
  has payoff $p$.

  Assume that $\sigma_\Agt$ is a normed strategy profile, which is
  resistant to immediately visible and honest single-player
  deviations, and which has payoff $p$.  Thanks to Lemma~\ref{lemma1},
  one can construct a Nash equilibrium $\widetilde\sigma_\Agt$, whose
  main outcome visits the same sequence of vertices as the main
  outcome of $\sigma_\Agt$, hence its payoff is $p$ as well. 
\end{proof}

\section{Correctness of the epistemic game construction}
\label{sec:correctness}

In this section, we let $\calG$ be a concurrent game, $G$ be a
communication graph for the players of $\calG$, and we let
$\calE_\calG^G$ be the corresponding epistemic game. We assume all
previous notations. 

We write $\ID = \{\id_d \mid d\in \Agt\}$ and
$\ID_\epsilon = \ID \cup \{\epsilon\}$.

\subsection{Basic properties of the epistemic game}

Considering an $\calE_{\calG}^G$-history $H =
(v_0,X_0),(v_0,X_0,f_0)...(v_k,X_k)$ (notice we always have perfect
alternation between \Eve's states and \Adam's state), we write
$\visited(H)$ for $v_0...v_k \in V^*$. We denote by
$\Hist_{\calE_\calG^G}(v_0,X_0)$ the set of such histories. An
important result that will be useful later is the following:

\lma*

\begin{proof}
 A proof by induction follows immediately from the structure of the
 epistemic game.
\end{proof}

\subsection{From \Eve's strategies in $\calE_\calG^G$ to normed profiles in
  $\calG$}

We fix a strategy $\zeta$ for \Eve, associating to every possible
$\calE_{\calG}^G$-history a certain \Eve's action.  From $\zeta$, one
defines inductively (on the size of histories) a partial function
$E_\zeta: \bigcup_{a\in \Agt} \Hist_{\calG,a}(v_0) \to
\Hist_{\calE_\calG^G}(v_0,\emptyset)$ associating a
$\calE_{\calG}^G$-history to some $a$-history $h$ in the original game
such that if $(m,\mes) \cdot v$ is a suffix of $h$, then 

\begin{itemize}
\item (i) for every $b \in \Vois(a)$, $\mes(b) \in \ID_\epsilon$;
\item (ii) for every $b,c \in \Vois(a)$, $\mes(b) \ne \mes(c)$ implies
  $\mes(b) = \epsilon$ or $\mes(c) = \epsilon$; for every
  $b \in \Vois(a)$, if the message sent by $b$ earlier along $h$ is
  $\id_d$ for some $d \in \Agt$, then $\mes(b) = \id_d$;
 \item (iii) $\last(E_\zeta(h)) = (v,X)$, with the
following additional properties: \label{ddag}
 \begin{itemize}
 \item $X = \emptyset$ implies for each $b \in \Vois(a)$, $\mes(b) =
   \epsilon$;
 \item if there is $b \in \Vois(a)$ with $\mes(b) = \id_d$, then $d
   \in \dev(X)$ and $a \in I^X_d$;
  \item if $X \ne \emptyset$ and for every $b \in \Vois(a)$,
   $\mes(b) = \epsilon$, then there exists $c \in \dev(X)$ such that
   $a \notin I^X_c$;
 \end{itemize}
\end{itemize}
  
We write $(\ddag)$ those conditions.

Together with this partial function $E_\zeta$, we define a strategy
profile $\sigma_\Agt$ in $\calG$, which will in some sense (that we
will explicit later) correspond to $\zeta$.

For history $v_0 \in \Hist_{\calG,a}(v_0)$, we abusively assume that
the last message (denoted $\mes$) assigns $\epsilon$ to every player in
$\Vois(a)$ (that is, $\mes(\Vois(a)) = \mes_\epsilon(\Vois(a))$). All
conditions $(\ddag)$ are then immediately satisfied and we define
$E_\zeta(v_0) = (v_0,\emptyset)$.

Pick now a history $h \in \Hist_{\calG,a}(v_0)$, ending with
$(m,\mes) \cdot v$, such that $E_\zeta(h)$ is well-defined (induction
hypothesis). Notice this implies that $h$ satisfies all conditions
$(\ddag)$.  Writing $\last(E_\zeta(h)) = (v,X)$, we will define
$\sigma_a(h)$, and extend $E_\zeta$ to several extensions of $h$. We
distinguish several cases:
\begin{itemize}
\item Assume first that $X=\emptyset$. Then,
  $\zeta(E_\zeta(h)) \in \Act^\Agt$. We set
  $\sigma_a(h) = \big(\zeta(E_\zeta(h))(a), \epsilon \big)$.

  We extend $E_\zeta$ as follows.
  \begin{itemize}
  \item First, if all players follow the suggestion of \Eve (that is,
    do not deviate), then the next state will be
    $v' = \Tab\big(v, \zeta(E_\zeta(h))\big)$. In this case, no
    message should be broadcast, and in $\calE_\calG^G$ (under
    $E_\zeta$), \Adam complies with the suggestion of \Eve and goes to
    $(v',\emptyset)$:
    \[
      E_\zeta\big(h \cdot \pi_a\big(\zeta(E_\zeta(h)), \mes_\epsilon
      \big) \cdot v'\big) = E_\zeta(h) \cdot ((v,\emptyset),
      \zeta(E_\zeta(h))) \cdot (v',\emptyset)
    \]
  \item Else, for every $d \in \Agt$ and $\delta \in \Act$, writing
    $v' = \Tab\big(v,\zeta(E_\zeta(h))[d/\delta]\big)$ and assuming
    that $v' = \Tab\big(v, \zeta(E_\zeta(h))\big)$ (which means it is
    an invisible deviation), no message should be broadcast as well
    since the deviation is somehow harmless, and in $\calE_\calG^G$
    (under $E_\zeta$), \Adam complies with the suggestion of \Eve and
    goes to $(v',\emptyset)$:
    \[
      E_\zeta\big(h \cdot \pi_a\big(\zeta(E_\zeta(h))[d/\delta],
      \mes_d \big) \cdot v'\big) = E_\zeta(h) \cdot
      ((v,\emptyset), \zeta(E_\zeta(h))) \cdot (v',\emptyset)
    \]
    where $\mes_d(d) = \id_d$ and $\mes_d(-d) = \mes_\epsilon(-d)$.
  \item Finally, for every $d \in \Agt$ and $\delta \in \Act$, writing
    $v' = \Tab\big(v,\zeta(E_\zeta(h))[d/\delta]\big)$ and assuming
    that $v' \ne \Tab\big(v, \zeta(E_\zeta(h))\big)$ (which means it
    is a visible deviation), and $X' =
    \update((v,\emptyset),\zeta(E_\zeta(h)),v')$ (which is then
    defined and nonempty):
    \[
      E_\zeta\big(h \cdot \pi_a\big(
      \zeta(E_\zeta(h))[d/\delta],\mes_d \big) \cdot v'\big) =
      E_\zeta(h) \cdot ((v,\emptyset), \zeta(E_\zeta(h))) \cdot
      (v',X')
    \]
    where $\mes_d(d) = \id_d$ and $\mes_d(-d) = \mes_\epsilon(-d)$.
  \end{itemize}
  Note that in all cases, the expected conditions are satisfied.
\item Assume that $X \ne \emptyset$.  If there is $b \in \Vois(a)$
  such that $\mes(b) \ne \epsilon$, then pick $d \in \dev(X)$ such
  that $\mes(b) = \id_d$; otherwise pick any $d \in \dev(X)$ such that
  $a \notin I^X_d$ (it must exist as well by induction
  hypothesis). We let $m = \zeta(E_\zeta(h))(d)$, and we set:
  \[
  \sigma_a(h) = \left\{\begin{array}{ll} (m(a),\id_d) &
      \text{if there is}\ b \in \Vois(a)\ \text{s.t.}\ \mes(b) = \id_d, \\
      (m(a),\epsilon) & \text{otherwise.}
    \end{array}\right.
  \]
  In the first case, player $a$ has received the message blaming
  player $d$, the deviator, while in the second case player $a$ 
  has received no message.
  
  This is well-defined for two different reasons: (i) those two cases
  can be distinguished in $h$ thanks to the messages
  $\pi_a(\mes) = \mes(\Vois(a))$; (ii) in the second case, thanks to
  Lemma~\ref{lma} and to the definition of actions in the epistemic
  game, the value $\zeta(E_\zeta(h))(d)(a)$ is independent of the
  choice of $d \in \dev(X)$ satisfying $a \notin I^X_d$.

  We furthermore extend $E_\zeta$ as follows.
  \begin{itemize}
  \item For every $\delta \in \Act$, writing
    $v' = \Tab(v,m[d/\delta])$, and
    $X' = \update((v,X),\zeta(E_\zeta(h)),v')$ (which is then
    defined and nonempty):
    \[    
      E_\zeta\big(h \cdot \pi_a(m[d/\delta],\mes') \cdot v'\big) =
      E_\zeta(h) \cdot ((v,X), \zeta(E_\zeta(h))) \cdot (v',X')
    \]
    with $\mes'(b) = \id_d$ if $b \in I^X_d$ and
    $\mes'(b) = \epsilon$ otherwise.  
  \end{itemize}
  Note that all conditions of the induction hypothesis are satisfied.
\end{itemize} 
Note that $(\ddag)$ are satisfied by every $h$ such that $E_\zeta(h)$
is defined.

This allows to define a function $\Omega$ associating a strategy
profile in the original game to every \Eve's strategy in the epistemic
game.

\medskip Pick $\zeta$ an \Eve's strategy in the epistemic game. Let
$\sigma = \Omega(\zeta)$ be the strategy profile in the original game
we have just constructed, together with the partial function
$E_\zeta$.  We show the following lemma:
\begin{lemma}
  \label{lemma5}
  \begin{itemize}
  \item Profile $\sigma$ is normed.
  \item Let $\rho$ be the main outcome of profile $\sigma$. Then,
    taking $E_\zeta$ at the limit, for every player $a \in \Agt$,
    $E_\zeta(\pi_a(\rho))$ is well-defined, and equal to $R$, the
    unique outcome of $\zeta$ where \Adam complies to \Eve's
    suggestions. In particular, $E_\zeta(\pi_a(\rho))$ is independent
    of the choice of player $a \in \Agt$.
  \item Consider a honest and immediately visible single-player
    deviation, and let $\rho'$ be the corresponding outcome. Then,
    taking $E_\zeta$ at the limit, for every $a \in \Agt$,
    $E_\zeta(\pi_a(\rho'))$ is well-defined, and equal to $R'$, an
    outcome of $\zeta$, where \Adam does not comply to \Eve's
    suggestions. In particular, $R' \ne R$.
  \end{itemize}
\end{lemma}

\begin{proof}
  Let
  $R = (v_0,\emptyset) \cdot ((v_0,\emptyset),f_0) \cdot
  (v_1,\emptyset) \dots (v_s,\emptyset) \dots$ be the outcome of
  $\zeta$, where \Adam complies to \Eve's suggestion. For every
  $r \ge 0$, $f_r \in \Act^\Agt$. We let
  $\rho = v_0 \cdot (f_0,\mes_\epsilon) \cdot v_1 \dots
  (f_{s-1},\mes_\epsilon) \cdot v_s \dots$. We can show (but omit it,
  since it is obvious) by induction on $s$ that (i) for every $a$,
  $E_\zeta(\pi_a(v_0 \cdot (f_0,\mes_\epsilon) \cdot v_1 \dots
  (f_{s-1},\mes_\epsilon) \cdot v_s))$ is defined, and is equal to
  $(v_0,\emptyset) \cdot ((v_0,\emptyset),f_0) \cdot (v_1,\emptyset)
  \dots (v_s,\emptyset)$, the prefix of length $s$ of $R$; and (ii)
  $\sigma_a(v_0 \cdot (f_0,\mes_\epsilon) \cdot v_1 \dots
  (f_{s-1},\mes_\epsilon) \cdot v_s) = ((f_s)_a,\epsilon)$. In
  particular, the unique outcome of profile $\sigma$ is $\rho$, and
  messages are handled correctly by $\sigma$ on that part.

  Then, notice that single-player deviations that do follow the
  sequence of vertices of the main outcome of $\sigma$ are considered
  as non-deviations by $\sigma$ (by construction, second item of the
  case $X = \emptyset$).

  Pick a `honest and visible' deviation $\sigma'_d$ for player $d$. We
  show that messages propagate properly, and that the outcome of
  $\sigma[d/\sigma'_d]$ corresponds to an outcome of $\zeta$, distinct
  from the main outcome, where \Adam complies to \Eve's
  suggestions. We write
  $\rho' = v'_0 \cdot (f'_0,\mes'_0) \cdot v'_1 \dots
  (f'_{s-1},\mes'_{s-1}) \cdot v'_s \dots$ for the outcome of
  $\sigma[d/\sigma'_d]$. There exists $s$ such that
  $v_0 \cdot (f_0,\mes_\epsilon) \cdot v_1 \dots
  (f_{s-1},\mes_\epsilon) \cdot v_s = v'_0 \cdot (f'_0,\mes'_0) \cdot
  v'_1 \dots (f'_{s-1},\mes'_{s-1}) \cdot v'_s$,
  $v'_{s+1} \ne v_{s+1}$ and $f'_s = f_s[d/\delta]$ for some action
  $\delta$, while $\mes'_s(d) = \id_d$ and $\mes'_s(a) = \epsilon$ if
  $a \ne d$ (this is by definition of a honest and visible
  deviation). So far, the message propagation system is working
  well. For every $r \ge s+1$, we write $h'_r$ for the prefix of
  length $r$ of $\rho'$. We show by induction on $r \ge s+1$ that for
  every $a$, $E_\zeta(\pi_a(h'_r))$ is well-defined and equal to
  $R'_r$, with $\last(R'_r) = (v'_r,X'_r)$ ($X'_r \ne \emptyset$),
  $R'_r$ is an outcome of $\zeta$, and the message propagation system
  has been working properly along $h'_r$.
  \begin{itemize}
  \item By construction of partial function $E_\zeta$, for every $a$,
    $E_\zeta(\pi_a(h'_{s+1}))$ is well-defined, and is equal to
    $R'_{s+1} = (v_0,\emptyset) \cdot ((v_0,\emptyset),f_0) \cdot
    (v_1,\emptyset) \dots (v_s,\emptyset) \cdot
    ((v_s,\emptyset),f_{s+1}) \cdot (v'_{s+1},X'_{s+1})$, where
    $X'_{s+1} = \update((v_s,\emptyset),f_{s+1},v'_{s+1})$ (with
    $X'_{s+1} \ne \emptyset$). Obviously, $R'_{s+1}$ is a prefix of an
    outcome of $\zeta$. So far, the message propagating system has
    worked well along $h'_{s+1}$.
  \item We assume that for every $a$, $E_\zeta(\pi_a(h'_r))$ (with
    $r \ge s+1$) is well-defined and equal to $R'_r$, with
    $\last(R'_r) = (v'_r,X'_r)$ ($X'_r \ne \emptyset$). We also assume
    that $R'_r$ is (a prefix of) an outcome of $\zeta$. Finally we
    assume that the message propagation system has worked properly
    along $h'_r$.

    We show that the same properties hold for $h'_{r+1}$. We fix a
    player $a$.
    Since $E_\zeta(\pi_a(h'_r)) = R'_r$, we have that $X'_r$ and
    $\mes'_{r-1}$ satisfy the conditions $(\ddag)$ given on
    page~\pageref{ddag}.

    To do the inductive step, we first look at
    $\sigma_a(h'_r)$:
    \[
      \sigma_a(h'_r) = \left\{\begin{array}{ll}
          \Big(\big(\zeta(E_\zeta(\pi_a(h'_r)))(d)\big)_a,\id_d\Big) &
          \text{if there is}\ b \in \Vois(a)\ \text{s.t.}\
                                                                       \mes'_{r-1}(b)
                                                                       \ne
                                                                       \epsilon \\
                                \Big(\big(\zeta(E_\zeta(\pi_a(h'_r)))(d)\big)_a,\epsilon\Big)
                                                                     &
                                                                       \text{otherwise}
      \end{array}\right.
    \]
    Hence, in all cases, we have
    $\sigma_a(h'_r)[1] = \big(\zeta(E_\zeta(\pi_a(h'_r)))(d)\big)_a =
    \big(\zeta(R'_r)(d)\big)_a$, and messages are properly propagated
    by player $a$.  We write $m' = \zeta(R'_r)(d)$, and define the
    message $\mes'_r$ by $\mes'_r(a) = \id_d$ if there is
    $b \in \Vois(a)$ such that $\mes'_{r-1}(b) = \id_d$, and
    $\mes'_r(a) = \epsilon$ otherwise.  Now, there is an action
    $\delta \in \Act$ such that
    \[
      h'_{r+1} = h'_r \cdot (m'[d/\delta],\mes'_r) \cdot v'_{r+1}
    \]
    We define 
    \[
    R'_{r+1} = R'_r \cdot ((v'_r,X'_r),\zeta(R'_r)) \cdot
    (v'_{r+1},X'_{r+1})
    \]
    with $X'_{r+1} = \update((v'_r,X),\zeta(R'_r),v'_{r+1})$.  Note
    that $X'_{r+1}$ is well-defined and nonempty since this is
    witnessed by
    $\Tab(v'_r,\big(\zeta(R'_r)(d)\big)[d/a]) = v'_{r+1}$. By
    construction of $E_\zeta$, we have that $E_\zeta(\pi_a(h'_{r+1}))$
    is well-defined and equal to $R'_{r+1}$.  Finally, along
    $h'_{r+1}$, the communication system has been working properly,
    and $R'_{r+1}$ is obviously (a prefix of) an outcome of $\zeta$
    (distinct from $R$). This concludes the proof of the inductive
    step, and of the lemma.
  \end{itemize}
\end{proof}

The following statement is an obvious consequence of the construction
of $E_\zeta$:

\begin{lemma}
  \label{sens1}
  If $\zeta$ is a winning strategy for \Eve in $\calE_\calG^G$, then
  $\sigma=\Omega(\zeta)$ is a normed strategy profile, which is
  resistant to single-player visible and honest deviations, and whose
  payoff is equal to the payoff of the outcome of $\zeta$ where \Adam
  complies to \Eve's suggestions.
\end{lemma}

\begin{proof}
  We assume that $\zeta$ is a winning strategy for \Eve, and that the
  payoff of the main outcome $R$ of $\zeta$ is $p = (p_a)_{a \in
    \Agt}$. Then, for each other outcome of $\zeta$, the payoff is
  bounded by $p$.
  
  Applying Lemma~\ref{lemma5}, the main outcome $\rho$ of $\sigma$ is
  such that $E_\zeta(\rho) = R$, yielding payoff $p$ for $\rho$.  Pick
  a honest and visible $d$-deviation $\sigma'_d$, and let $\rho'$ be
  the outcome of $\sigma[d/\sigma'_d]$. Then, $E_\zeta(\rho')$ is
  defined and is an outcome of $\zeta$ (again by application of
  Lemma~\ref{lemma5}), which payoff is therefore bounded by $p$. Hence, 
  $\sigma$ satisfies the expected conditions.
\end{proof}

\subsection{From normed profiles in $\calG$ to \Eve's strategies in
  $\calE_\calG^G$}

We assume an arbitrary total order $<$ on the set $\Act$. This will be
used to define unique corresponding (local) histories in $\calG$.

We first define a mapping assigning families of (local) histories in
$\calG$ to histories in $\calE_\calG^G$. Consider an
$\Eve$'s history
$H = (v_0,X_0) \cdot (v_0,X_0,f_0) \dots (v_s,X_s)$ in
$\calE_\calG^G$.
\begin{itemize}
\item If $X_s = \emptyset$, then for every $r < s$,
  $X_r = \emptyset$ as well, and therefore $f_r \in \Act^\Agt$. We
  then associate with $H$ the single full history, which is easily
  seen to be well-defined:
  \[
    h=v_0 \cdot (f_0,\mes_\epsilon) \cdot v_1 \cdot
    (f_1,\mes_\epsilon) \cdot v_2 \dots (f_{s-1},\mes_\epsilon) \cdot
    v_s
  \]
\item If $X_s \ne \emptyset$, then there is a smallest index $0<r_0
  \le s$ such that $X_{r_0} \ne \emptyset$, and for every $r_0 \le r
  \le s$, $X_r \ne \emptyset$. Note also that for every $r_0 \le r_1
  \le r_2 \le s$, $\dev(X_{r_2}) \subseteq \dev(X_{r_1})$.  We then
  associate with $H$ and with every deviator $d \in \dev(X_s)$, the
  (unique) full history:
  \[
    h_d = v_0 \cdot (m_0,\mes_0) \cdot v_1 \dots (m_{s-1},\mes_{s_1})
    \cdot v_s
  \]
  such that:
  \begin{itemize}
  \item $l < r_0$ implies $m_l = f_l$ and $\mes_l = \mes_\epsilon$;
  \item Let $l \ge r_0$. For every $a \in \Agt$, $\mes_l(a) = \id_d$
    if $a \in I^{X_l}_d$, otherwise $\mes_l(a) = \epsilon$. Then,
    $m_l(-d) = f_l(d)(-d)$, and $m_l(d)$ is the $<$-smallest
    action $\delta \in \Act$ such that $v_{l+1} = \Tab(v_l,
      f_l(d)[d/\delta])$.
  \end{itemize}
\end{itemize}
We denote by $\Lambda$ the function that associates with every
$\Eve$'s history $H$ in $\calE_\calG^G$, either the single full 
history $h$ (first case), or the family of full histories 
$(h_d)_{d \in \dev(X_s)}$ (second case).

\medskip Let $\sigma$ be a normed strategy profile in the original
game. We define the \Eve's strategy $\zeta$ as follows:
\begin{itemize}
\item if $H$ is such that $\last(H) = (v,\emptyset)$, then $\zeta(H)
  \in \Act^\Agt$ and $\zeta(H)(a) = \sigma_a(\Lambda(H))[1]$;
\item if $H$ is such that $\last(H) =(v,X)$ with $X \ne \emptyset$,
  then $\zeta(H) : \dev(X) \to \Act^\Agt$ is such that
  $\zeta(H)(d)(a)= \sigma_a(\Lambda(H)(d))[1]$ for all $d \in
  \dev(X)$.
\end{itemize}
This allows to define a function $\Upsilon$ associating a strategy in
the epistemic game to every normed profile in the original game.

Pick $\sigma$ a normed strategy profile in $\calG$, and write $\zeta =
\Upsilon(\sigma)$ for the corresponding strategy in $\calE_\calG^G$.

\begin{lemma}
  \label{lemma7}
  \begin{itemize}
  \item Let $R$ be the unique outcome of $\zeta$ where \Adam complies
    to \Eve's suggestions. Then, at the limit, $\Lambda(R)$ is the
    unique outcome $\rho$ of profile $\sigma$.
  \item Let $R'$ be an outcome of $\zeta$ along which \Adam does not
    always comply to \Eve's suggestions. Then, for every $d \in
    \lim_{s \to +\infty} \dev(X'_s)$, there exists some honest and
      visible deviation $\sigma'_d$ such that $\Lambda(R')(d) =
      \out(\sigma[d/\sigma'_d],v_0)$.
  \end{itemize}
\end{lemma}

\begin{proof}
  Write
  $R = (v_0,\emptyset) \cdot ((v_0,\emptyset),f_0) \cdot
  (v_1,\emptyset) \cdot ((v_1,\emptyset),f_1) \cdot (v_2,\emptyset)
  \dots$ for the outcome of $\zeta$, along which \Adam complies to
  \Eve's suggestions.  Then it is easy to argue that the outcome of
  $\sigma$ is
  $\rho = v_0 \cdot (f_0,\mes_\epsilon) \cdot v_1 \cdot
  (f_1,\mes_\epsilon) \cdot v_2 \dots$, and $\rho$ coincides with
  $\Lambda(R)$ (at the limit).

  Let $R' = (v'_0,X'_0) \cdot ((v'_0,X'_0),f'_0) \cdot (v'_1,X'_1)
  \dots$ such that there is some (the first) $r$ such that $X'_r \ne
  \emptyset$. We show by induction on $s \ge r$ that for every $d \in
  \dev(X'_s)$, there is a honest $d$-deviation $\sigma'_d$ such that
  the length-$s$ outcome of $\sigma[d/\sigma'_d]$ is
  \[
    v'_0 \cdot (f'_0,\mes_0) \cdot v'_1 \cdot (f'_1,\mes_1) \cdot v'_2
    \dots (f'_{s-1},\mes_{s-1}) \cdot v'_s = \Lambda(R'_{\le s})(d)
  \]
  In the same induction, we will prove that
  $I^{X'_s}_d = \{a \in \Agt \mid \exists b \in \Vois(a)\ \text{s.t.}\
  \mes_{s-1}(b) = \id_d\}$.

  Before proving the induction, notice that
  $v'_0 \cdot (f'_0,\mes_0) \cdot v'_1 \cdot (f'_1,\mes_1) \dots
  v'_{r-1} = v_0 \cdot (f_0,\mes_\epsilon) \cdot v_1 \cdot
  (f_1,\mes_\epsilon) \dots v_{r-1}$ (that is, it follows the main
  outcome), and $v'_r \ne v_r$.
  \begin{itemize}
  \item Assume $s=r$. Then,
    $X'_r = \update((v_{r-1},\emptyset),\zeta(R_{\le r-1}),v'_r)$
    (defined and nonempty), with
    $\zeta(R_{\le r-1})(a) = \sigma_a(\rho_{\le r-1})[1]$. By
    construction, for every $d \in \dev(X'_r)$, there is
    $\delta \in \Act$ such that
    $v'_r = \Tab(v_{r-1},\sigma(\rho_{\le r-1})[d/\delta])$, where
    $\sigma(\rho_{\le r-1})$ denotes abusively the tuple
    $(\sigma_a(\rho_{\le r-1})[1])_{a \in \Agt}$. We define
    $\sigma'_d(\rho_{\le r-1}) = (\delta,\id_d)$ (this is a honest and
    visible deviation). Then,
    $\rho_{\le r} \in \out(\sigma[d/\sigma'_d],v_0)$. Also, since the
    system of message propagation under $\sigma$ is behaving well, it
    is the case that for every $a$, $\mes_{s-1}(a) = \epsilon$ if
    $a \ne d$, and $\mes_{s-1}(d) = \id_d$ (as given by
    $\sigma'_d(\rho_{\le r-1})$). It is an easy check to prove the
    property on messages.
  \item Assume we have proven the result at rank $s$, and consider
    $R'_{\le s+1} = R'_{\le s} \cdot ((v'_s,X'_s),f'_s) \cdot
    (v'_{s+1},X'_{s+1})$. Pick $d \in \dev(X'_{s+1})$. Since then,
    $d \in \dev(X'_s)$ as well, we can apply the induction hypothesis
    to $R'_{\le s}$, and we have a deviation $\sigma'_d$ such that the
    length-$s$ outcome of $\sigma[d/\sigma'_d]$ is
    $\Lambda(R'_{\le s})(d)$.

    By definition, $X'_{s+1} = \update((v'_s,X'_s),f'_s,v'_{s+1})$,
    hence there exists $\delta \in \Act$ such that
    $\Tab(v'_s,f'_s(d)[d/\delta]) = v'_{s+1}$. We then define
    $\sigma'_d(\Lambda(R'_{\le s})_d) = (\delta,\id_d)$. The
    length-$(s+1)$ outcome of $\sigma[d/\sigma'_d]$ is
    \[
    \Lambda(R'_{\le s})(d) \cdot (m[d/\delta],\mes) \cdot v'
    \]
    where:
    \begin{itemize}
    \item for every $a \in \Agt$, $m(a) = \sigma_a(\Lambda(R'_{\le
        s})(d))[1]$
    \item for every $a \in \Agt$,
      $\mes(a)= \sigma_a(\Lambda(R'_{\le s})(d))[2]$
    \item $v' = \Tab(v'_s,m[d/\delta])$
    \end{itemize}
    Now, we show $m[d/\delta] = f'_s(d)$: $f'_s(d)(a) = (\zeta(R'_{\le
      s})(d))(a) = \sigma_a(\Lambda(R'_{\le s})(d)[1] = m(a)$. Hence,
    $v' = v'_{s+1}$. We conclude that the length-$(s+1)$ outcome of
      $\sigma[d/\sigma'_d]$ is
    \[
    \Lambda(R'_{\le s})(d) \cdot (\zeta(R'_{\le
      s})(d)[d/\delta],\mes) \cdot v'_{s+1} = \Lambda(R'_{\le
      s+1})(d)
    \]
    Concerning the messages:
    $I^{X'_{s+1}}_d = \{a \in \Agt \mid \exists b \in I^{X'_s}_d\
    \text{s.t.}\ \dist_G(a,b) \le 1\}$. Since $d$ is the deviator and
    the deviation is honest, $\mes_{s-1}(d) = \mes_s(d) = \id_d$, and
    for every $a \ne d$,
    $\mes_s(a) = \mes(a) = \sigma_a(\Lambda(R'_{\le s})(d))[2]$. Since
    $\sigma$ is normed, messages are propagated correctly by all players
    $a \ne d$. Hence the expected equality holds.
   \end{itemize}
\end{proof}

\begin{lemma}
  \label{sens2}
  If $\sigma$ is a normed strategy profile in $\calG$, which is
  resistant to single-player immediately visible and honest
  deviations, then $\zeta = \Upsilon(\sigma)$ is a winning strategy in
  $\calE_\calG^G$.
\end{lemma}

\begin{proof}
  Let $p$ be the payoff associated with $\rho = \out(\sigma,v_0)$. The
  outcome of $\zeta$ when \Adam complies to \Eve's suggestions is $R$
  such that $\rho = \Lambda(R)$. In particular, $R$ has the same
  payoff as $\rho$, that is, $p$.  

  Assume now that $R' = (v_0,X_0) \cdot ((v_0,X_0),f_0) \cdot
  (v_1,X_1) \dots$ is a play in the epistemic game such that from some
  point on, $X_i \ne \emptyset$. Then, from Lemma~\ref{lemma7}, for
  every $d \in \lim_{s \to +\infty} \dev(X'_s)$, there exists
    $\sigma'_d$ such that
  \[
  \rho' = \out(\sigma[d/\sigma'_d],v_0) = \Lambda(R')(d)
  \]
  The payoff of $\rho'$ and $R'$ therefore coincide (and are equal to
  $p'$). Since $\sigma$ is a Nash equilibrium, $p'_d \le p_d$. Hence
  for every $d \in \lim_{s \to +\infty} \dev(X'_s)$, the payoff of
  player $d$ along $R'$ is bounded by $p_d$. Hence $\zeta$ is
  winning. 
\end{proof}

\subsection{Conclusion}

As a consequence of Theorem~\ref{coro}, Lemmas~\ref{sens1}
and~\ref{sens2}, we get Theorem~\ref{theo:main}:

\main*

Note that all the results are constructive, hence if one can
synthesize a winning strategy for \Eve in $\calE_\calG^G$, then we can
synthesize a correponding Nash equilibrium in $\calG$.

\section{Complexity analysis}

By application of Lemma~\ref{lma}, we get:

\begin{lemma}
  \label{coro:small}
  Let $(v,\emptyset) \cdot ((v,\emptyset),f_0) \cdot (v_1,X_1) \cdot
  ((v_1,X_1),f_1) \cdot (v_2,X_2) \dots$ with $X_1 \ne \emptyset$ be a
  history in $\calE_\calG^G$. Then for every $r \ge 1$:
  \begin{itemize}
  \item $\dev(X_r) \subseteq \dev(X_1)$;
  \item for every $d \in \dev(X_r)$, for every $a \in \Agt$,
    $\dist_G(d,a) \le r$ iff $a \in I_d^{X_r}$. 
  \item If $d \in \dev(X_r)$ and $\dist_G(d,a) \le r$, then
    $K_d^{X_r}(a) = \{d\}$.
  \item If $d \in \dev(X_r)$ and $\dist_G(d,a) > r$, then
    $K_d^{X_r}(a) = \dev(X_r) \setminus \{d' \in \dev(X_r) \mid
    \dist_G(d',a) \le r\}$.
  \end{itemize}
\end{lemma}

\size*

\begin{proof}
  We start by evaluating the number of \Eve's states.
  First, the number of \Eve's states $(v,\emptyset)$ is obviously
  $|V|$.

  Then, pick an \Eve state $(v',X')$ with $X' \ne \emptyset$, such
  that there is a transition $((v,\emptyset),f) \to (v',X')$ in
  $\calE_\calG^G$ (those are immediate visible deviations). Then,
  following an argument used in~\cite[Prop.~4.8]{BBMU15}, we can show
  that $|\Tab| \ge 2^{|\dev(X')|}$: indeed, each player in $\dev(X')$
  has been able to deviate, hence can at least do two actions from the
  current state, yielding the claimed bound.

  We will now analyze the part which is reachable from
  $(v',X')$. Applying Lemma~\ref{coro:small}, any \Eve's state
  $(v'',X'')$ reachable from $(v',X')$ is such that $\dev(X'')
  \subseteq \dev(X')$, and is fully characterized by
  $(v',\dev(X''),r)$ where $\dev(X'') \subseteq \dev(X')$ and $r$ is
  the distance from $(v',X')$. Hence, this number of states is bounded
  by $|V| \cdot 2^{|\dev(X')|} \cdot (\diam(G)+2)$, where $\diam(G)$
  is the maximal diameter of the connected components of $G$ (the $+2$
  term is for ``distance $+\infty$'' and for ``distance larger than
  the diameter''). Hence it is bounded by $|V| \cdot |\Tab| \cdot
  (\diam(G)+2)$.

  Since there are at most $|\Tab|$ possible deviation starting points
  $(v',X')$, the number of \Eve's states is bounded by $|V| + |V| \cdot
  |\Tab|^2 \cdot (\diam(G)+2)$.

  Now we evaluate the number of \Adam's states. There are at most
  $|\Tab|$ states of the form $((v,\emptyset),f)$. Now, from an
  \Eve's state $(v,X)$ with $X \ne \emptyset$, there are \Adam's states
  $((v,X),f)$ with $f : \dev(X) \to \Act^\Agt$. It is a priori
  difficult to reduce the number of such $f$, which is bounded by
  $|\Act|^{|\dev(X)| \cdot |\Agt|}$, hence by~$|\Act|^{|\Agt|^2}$.
  %
  %
  %
\end{proof}

\end{document}